\documentclass[11pt]{article}
\usepackage{latexsym}
\usepackage{amssymb}
\usepackage{bussproofs}
\usepackage{pstricks}
\usepackage{color}

\newtheorem{theorem}{Theorem}

\newtheorem{lemma}[theorem]{Lemma}

\newenvironment{definition}{\begin{trivlist}\item[]{\bf Definition}\ }%
{\end{trivlist}}
\newenvironment{proof}{\begin{trivlist}\item[]{\bf Proof\ }}%
{\end{trivlist}}

\def\lorunary{{\mathop{\bigvee}}}           
\def\falseBPO#1{{\left(\lorunary{#1}\right)}}
\def\falseBPOpi{{\falseBPO{\overline\pi}}}
\def\falseBPOpisub#1{{\falseBPO{\overline\pi_{#1}}}}

\newcommand{\qed}{\Box}			

\newcommand{\rddots}{{\mathinner{\mkern1mu\raise1pt\vbox{\kern7pt\hbox{.}}%
        \mkern2mu\raise4pt\hbox{.}\mkern2mu\raise7pt\hbox{.}\mkern1mu}}}
\newcommand{\proofdots}{{\ddots\vdots\,\rddots}}

\def\pprime{{\prime\prime}}

\def\alphabeta{{\alpha\beta}}
\def\negA{{\overline A}}
\def\negB{{\overline B}}
\def\negE{{\overline E}}
\def\negF{{\overline F}}
\def\negT{{\overline T}}
\def\negU{{\overline U}}
\def\negV{{\overline V}}
\def\negW{{\overline W}}
\def\negx{{\overline x}}
\def\negy{{\overline y}}

\def\koplus{{k\oplus}}
\def\qmark{{\hbox{\bfseries \sf ?}}}

\def\GT{{\mbox{\textup{GT}}}}
\def\GGT{{\mbox{\textup{GGT}}}}
\def\Peb{{\mbox{\textup{Peb}}}}
\def\GPeb{{\mbox{\textup{GPeb}}}}

\def\caseiti{{(\kern-1pt{\it i})}}
\def\caseitii{{(\kern-1pt{\it ii})}}
\def\caseitiii{{(\kern-1pt{\it iii})}}
\def\caseitiv{{(\kern-1pt{\it iv})}}
\def\caseitip{{(\kern-1pt{\it i'})}}
\def\caseitiip{{(\kern-1pt{\it ii'})}}
\def\caseitiiip{{(\kern-1pt{\it iii'})}}
\def\caseitivp{{(\kern-1pt{\it iv'})}}

\def\rest{{\upharpoonright}}

\begin{document}

\title{Improved Separations of Regular Resolution
from Clause Learning Proof Systems  \\
{\em \normalsize Preliminary version. Comments appreciated.}}

\author{
Maria Luisa Bonet\thanks{Supported in part by grant TIN2010-20967-C04-02.}\\
\small Lenguajes y Sistemas Inform{\'a}ticos \\
\small Universidad Polit{\'e}cnica de Catalu{\~n}a \\
\small Barcelona, Spain \\
\small \tt bonet@lsi.upc.edu
\and
Sam Buss\thanks{Supported in part by
NSF grants \hbox{DMS-0700533} and \hbox{DMS-1101228}, and by a grant
from the Simons
Foundation (\#208717 to Sam Buss).
The second author thanks the John Templeton Foundation for supporting
his participation in the CRM Infinity Project
at the Centre de Recerca Matem\`atica,
Barcelona, Catalonia, Spain during which
some of these results were obtained.}~${}^\ddag$\\
\small Department of Mathematics \\
\small University of California, San Diego\\
\small La Jolla, CA 92093-0112, USA\\
\small \tt sbuss@math.ucsd.edu
\and
Jan Johannsen\thanks{We thank the Banff International Research Station
for a workshop on proof complexity
held October 2011 during which part of these results were obtained.}  \\
{\small Institut f\"ur Informatik} \\
{\small Ludwig-Maximilians Universit\"at} \\
{\small D-80538 M\"unchen, Germany } \\
{\small \tt jan.johannsen@ifi.lmu.de }
}

\maketitle

\begin{abstract}
We prove that the graph tautology formulas of
Alekhnovich, Johannsen, Pitassi, and Urquhart
have polynomial size pool resolution refutations that use
only input lemmas as learned clauses and without degenerate
resolution inferences.  We also prove that these
graph tautology formulas can be refuted
by polynomial size DPLL proofs with
clause learning, even when restricted to greedy,
unit-propagating DPLL
search.  We prove similar results for the
guarded, xor-fied pebbling tautologies which
Urquhart proved are hard for regular resolution.
\end{abstract}

\section{Introduction}

The problem SAT of deciding the satisfiability of propositional CNF formulas
is of great theoretical and practical interest. Even though it is
NP-complete, industrial instances with hundreds of thousands variables
are routinely solved
by state of the art SAT solvers.  Most of these solvers are
based on the DPLL procedure~\cite{DLL:theoremproving}
extended with clause learning, restarts, variable
selection heuristics, and other techniques.

The basic DPLL procedure
without clause learning is equivalent to tree-like
resolution.  The addition of clause learning makes DPLL considerably
stronger.  In fact, clause learning together with unlimited
restarts is capable of simulating general resolution
proofs~\cite{PipatsrisawatDarwiche:clauselearning}.
However, the
exact power of DPLL with clause learning but without restarts
is unknown.  This question
is interesting not only for theoretical
reasons, but also because of the potential for better understanding
the practical performance of
various refinements of DPLL with clause learning.

Beame, Kautz, and Sabharwal~\cite{BKS:clauselearning} gave
the first
theoretical analysis of DPLL
with clause learning.
Among other things,
they noted that clause learning with restarts simulates
general resolution.  Their construction required the
DPLL algorithm to ignore some contradictions,
but this situation was rectified
by Pipatsrisawat and
Darwiche~\cite{PipatsrisawatDarwiche:clauselearning} who showed
that SAT solvers which do not ignore contradictions can also
simulate resolution.   These techniques were also applied
to learning bounded width clauses by Atserias~et~al.~\cite{AFT:clauseLearning}.

Beame et al.~\cite{BKS:clauselearning} also studied DPLL clause
learning without restarts.  Using a method of ``proof trace extensions'',
they were able to show that DPLL with clause learning
and no restarts is strictly
stronger than any ``natural'' proof system strictly weaker than
resolution.  Here, a {\em natural} proof system is one in which
proofs do not increase in length substantially when variables are restricted
to constants.  The class of natural proof systems is known to
include common proof systems such as tree-like or regular proofs.
The proof trace method involves introducing extraneous variables
and clauses, which have the effect of giving
the clause learning DPLL algorithm more freedom in choosing
decision variables for branching.

There have been two approaches to formalizing DPLL with clause
learning as a static proof system rather
than as a proof search algorithm.  The first is
pool resolution with a degenerate resolution inference, due
originally to Van Gelder~\cite{VanGelder:PoolResolution} and
studied further by Hertel et al.~\cite{HBPvG:clauselearn}.
Pool resolution requires proofs to have a depth-first regular
traversal similarly to the search space of a DPLL algorithm.
Degenerate resolution allows resolution inferences in which one or both of
the hypotheses may be lacking occurrences of the resolution
literal.  Van Gelder argued that pool resolution with degenerate
resolution inferences simulates a wide range of DPLL algorithms
with clause learning.  He also gave a proof,
based on~\cite{AJPU:regularresolution}, that pool
resolution with degenerate inferences is stronger than
regular resolution, using extraneous
variables similar to proof trace extensions.

The second approach is due to
Buss~et~al.~\cite{BHJ:ResTreeLearning}
who introduced a ``partially degenerate'' resolution
rule called w-resolution, and a proof system regWRTI
based on w-resolution and clause learning of
``input lemmas''.  They proved that
regWRTI exactly captures non-greedy DPLL with clause learning.
By ``non-greedy'' is meant that contradictions may need to be
ignored by the DPLL search.

Both \cite{HBPvG:clauselearn} and~\cite{BHJ:ResTreeLearning}
gave improved versions of the proof trace extension method
so that the extraneous variables depend only on the
set of clauses being refuted and not on resolution
refutation of the clauses.
The drawback remains, however, that the proof trace extension
method gives contrived sets of clauses and contrived resolution
refutations.

It remains open whether any of
DPLL with clause learning,
pool resolution (with or without
degenerate inferences),
or the regWRTI proof system can polynomially
simulate general resolution.
One approach to answering
these questions is to try to separate pool resolution (say)
from general resolution.
So far, however, separation results are
known only for the weaker system of regular
resolution, based on work of
Alekhnovich et al.~\cite{AJPU:regularresolution},
who gave an exponential separation between
regular resolution and general resolution.
Alekhnovich et al.~\cite{AJPU:regularresolution}
proved this separation for two families
of tautologies, variants of the graph tautologies~$\GT^\prime$ and
the ``Stone'' pebbling tautologies.
Urquhart~\cite{Urquhart:regularresolution} subsequently
gave a related separation using a different set
of pebbling tautologies which he denoted denoted~$\Pi_i$.\footnote{%
Huang and Yu~\cite{HuangYu:regularresolution} also gave a
separation of regular resolution and general resolution, but
only for a single set of clauses.
Goerdt~\cite{Goerdt:regularresolution}
gave a quasipolynomial separation of regular resolution and
general resolution.}
In the present paper, we call the tautologies~$\GT^\prime$
the {\em guarded} graph tautologies, and henceforth denote
them $\GGT$ instead of $\GT^\prime$; their definition is given
in Section~\ref{PrelimsSect}.
We define the formulas $\GPeb^{\koplus}(G)$ in Section~\ref{sec:GPeb};
these are essentially the $\Pi_i$ tautologies
of Urquhart.

An obvious question is whether
pool resolution (say) has polynomial size proofs of
the $\GGT$ tautologies, the $\GPeb^{\koplus}$, or the Stone tautologies.
The present paper resolves the first two questions
by showing that pool resolution does indeed have polynomial
size proofs of the graph tautologies~$\GGT$ and the
pebbling tautologies $\GPeb^{\koplus}$.
Our refutations avoid the use of extraneous variables in the style
of proof trace extensions;
furthermore, they use only the traditional
resolution rule and do not require
degenerate resolution inferences or w-resolution
inferences.   In addition, we use only learning of input clauses;
thus, our refutations are also regWRTI
proofs (and in fact regRTI proofs) in the terminology
of~\cite{BHJ:ResTreeLearning}.
As a corollary of the characterization of regWRTI
by~\cite{BHJ:ResTreeLearning},
the $\GGT$ principles and the $\GPeb^{\koplus}$ principles
have polynomial size refutations that can be found by
a DPLL algorithm with clause learning and without restarts (under
the appropriate variable selection order).

The Stone principles have recently been shown to also
have regRTI refutations by the second author and
L.~Ko{\l}odziejczyk~\cite{BussKolodziejczyk:SmallStone};
however, their
proof uses a rather different method than we use below.
Thus, none of the three principles separate clause learning
DPLL from full resolution.
It is natural to speculate that
perhaps pool resolution or regWRTI can simulate
general resolution, or that DPLL with clause learning and without
restarts can simulate general resolution.  It is far from clear
that this is true, but, if so, our methods and
those of~\cite{BussKolodziejczyk:SmallStone} may
represent a step in this direction.

The outline of the paper is as follows.  Section~\ref{PrelimsSect}
begins with the definitions of resolution, degenerate resolution, and
w-resolution, and then regular, tree, and pool resolution.
After that, we define the
graph tautologies $\GT_n$ and the guarded versions $\GGT_n$,
and state the main theorems about proofs
of the $\GGT_n$ principles.  Section~\ref{MainPfSect} gives
the proof of these theorems.  Several ingredients are
needed for the proof.
The first idea is to try to
follow the regular refutations of the
graph tautology clauses $\GT_n$ as given by
St{\aa}lmarck~\cite{Stalmarck:trickyformulas}
and Bonet and Galesi~\cite{BonetGalesi:ResAndPolyCalculus}:
however, these refutations
cannot be used directly since the transitivity clauses
of $\GT_n$ are ``guarded'' in the $\GGT_n$ clauses
and this yields refutations which
violate the regularity/pool property.  So, the
second idea is that
the proof search process branches as needed to learn
transitivity clauses.  This generates additional clauses
that must be proved: to handle these, we develop a notion
of ``bipartite partial order'' and show that the
refutations of
\cite{Stalmarck:trickyformulas,BonetGalesi:ResAndPolyCalculus}
can still be used in the presence of a bipartite partial order.
The tricky part is to be sure that exactly the right set of
clauses is derived by each subproof.
Some straightforward bookkeeping
shows that the resulting proof is polynomial size.

Section~\ref{GreedySect} discusses how to modify
the refutations constructed in Section~\ref{MainPfSect}
so that they are ``greedy'' and ``unit-propagating''.
These conditions means that proofs cannot ignore contradictions,
nor contradictions that can be obtained by unit propagation.
The
greedy and unit-propagating conditions correspond well to actual
implemented DPLL proof search algorithms, since they backtrack whenever
a contradiction can be found by unit propagation.  Section~\ref{GreedySect}
concludes with an explicit description of a polynomial time
DPLL clause learning
algorithm for the $\GGT_n$ clauses.

Section~\ref{sec:GPeb} gives the pool resolution and
regRTI refutations of the $\GPeb^\koplus$ principles.  The proof
mimics to a certain extent the methods of Section~\ref{MainPfSect},
but must also deal with the complications of xor-ification.

This paper is an expansion of an
extended abstract~\cite{BonetBuss:poolVsRegularSAT} and
an unpublished preprint~\cite{BonetBuss:poolVsRegularArxiv} by the
first two authors.
These earlier versions included only the results
for the $\GGT$ tautologies
and did not consider the $\GPeb$ principles.

We are grateful to J.\ Hoffmann for assisting with a correction
to an earlier version of the proof of Theorem~\ref{regRtiGgtThm}.
We also thank A.\ Van Gelder, A.\ Beckmann, and T.\ Pitassi for
encouragement, suggestions, and useful comments.

\section{Preliminaries and guarded graph tautologies}\label{PrelimsSect}

Propositional formulas are defined over a
set of variables and the connectives $\wedge$, $\vee$ and $\neg$.
We use the notation $\negx$
to express the negation~$\lnot x$ of~$x$. A {\em literal} is either
a variable~$x$ or a negated variable~$\negx$.
A {\em clause}~$C$ is a set of literals, interpreted
as the disjunction of its members.
The empty clause,~$\Box$, has truth value {\em False}.
We shall only use formulas
in {\em conjunctive normal form}, CNF; namely, a formula will be
a set (conjunction) of clauses.  We often use disjunction~($\lor$),
union~($\cup$), and comma ($,$) interchangeably.

\begin{definition}
The various forms of resolution
take two clauses $A$ and $B$ called the {\em premises}
and a literal $x$
called the {\em resolution variable},
and produce a new clause $C$ called the {\em resolvent}.
\begin{prooftree}
\AxiomC{A} \AxiomC{B}
\BinaryInfC{C}
\end{prooftree}
In all cases below, it is required that $\negx \notin A$ and
$x\notin B$.
The different forms of resolution are:
\begin{description}
\item {\em Resolution rule.}  The hypotheses have the forms
$A := A'\lor x$ and $B:=B'\lor \, \overline{x}$.
The resolvent~$C$ is $A'\lor B'$.
\item {\em Degenerate resolution rule.} \cite{HBPvG:clauselearn,VanGelder:PoolResolution}
If $x\in A$ and $\overline{x}\in B$,
we apply the resolution rule to obtain~$C$.
If $A$ contains~$x$, and $B$ doesn't contain~$\overline x$,
then the resolvent $C$ is $B$.
If $A$ doesn't contain~$x$,
and $B$ contains~$\overline{x}$,
then the resolvent~$C$ is~$A$.
If neither $A$ nor~$B$ contains the literal $x$ or~$\overline x$,
then $C$ is the lesser of $A$ or~$B$ according to some tiebreaking
ordering of clauses.
\item{\em w-resolution rule.} \cite{BHJ:ResTreeLearning}
From $A$ and~$B$ as above, we infer
$C:=(A\setminus\{x\})\lor(B\setminus\{\overline{x}\})$.
If the literal $x\notin A$ (resp., $\overline x\notin B$), then it
is called a {\em phantom literal} of $A$ (resp.,~$B$).
\end{description}
\end{definition}

\begin{definition}
A {\em resolution derivation}, or {\em proof}, of a clause~$C$
from a CNF formula~$F$
is a sequence of clauses $C_1,\ldots,C_s$ such that $C=C_s$
and such that each clause
from the sequence is either a clause from~$F$ or
is the resolvent of two previous clauses.
If the derived clause,~$C_s$, is the empty clause, this is
called a {\em resolution refutation} of~$F$.
The more general systems of
degenerate and w-resolution refutations are defined similarly.
\end{definition}

We represent a derivation
as a directed acyclic graph (dag) on the
vertices $C_1,\ldots,C_s$,
where each clause from~$F$ has out-degree~$0$,
and all the other vertices from $C_1,\ldots,C_s$ have edges pointing
to the two clauses from which they were derived.
The empty clause has in-degree~$0$.
We use the terms ``proof'' and ``derivation'' interchangeably.

Resolution is sound and complete
in the refutational sense: a CNF
formula~$F$ has a refutation if
and only if $F$ is unsatisfiable, that is, if and only
if $\lnot F$ is a tautology.
Furthermore, if there is a derivation of a clause~$C$
from~$F$,
then $C$ is a consequence of~$F$;
that is, for every truth assignment~$\sigma$,
if $\sigma$ satisfies $F$ then it satisfies~$C$.
Conversely, if $C$ is a consequence of $F$ then
there is a derivation
of some $C'\subseteq C$ from~$F$.

A resolution refutation is {\em regular} provided that,
along any path in the directed acyclic graph,
each variable is resolved on at most once.  A resolution
derivation of a clause~$C$ is {\em regular} provided
that, in addition, no variable appearing in~$C$ is used as
a resolution variable in the derivation.
A refutation is {\em tree-like} if the underlying graph is a tree;
that is, each occurrence of a clause occurring in the refutation
is used at most once as a
premise of an inference.

We next define a version of pool resolution,
using the conventions of~\cite{BHJ:ResTreeLearning} who called
this ``tree-like regular resolution with lemmas'' or ``regRTL''.
The idea is that clauses obtained previously in the proof
can be used freely as learned lemmas.
To be able to talk about clauses previously obtained,
we need to define an ordering of clauses.

\begin{definition}
Given a tree $T$, the {\em postorder}
ordering $<_T$ of the nodes is defined as follows:
if $u$ is a node of~$T$,
$v$~is a node in the subtree rooted at the left child of~$u$,
and $w$~is a node in the subtree rooted at the right child
of~$u$, then $v<_T w<_T u$.
\end{definition}

\begin{definition}
A {\em pool resolution} proof (also called a regRTL proof)
from a set of initial clauses~$F$
is a resolution proof tree~$T$
that fulfills the following conditions:
(a)~each leaf is labeled with either a clause of~$F$ or a clause
(called a ``lemma'')
that appears earlier in the tree in the $<_T$ ordering;
(b)~each internal node is labeled with a clause and a literal,
and the clause is obtained by resolution
from the clauses labeling the node's children
by resolving on the given literal;
(c)~the proof tree is regular;
(d)~the roof is labeled with the conclusion clause.
If the labeling of the root is the empty
clause $\Box$, the pool resolution proof
is a {pool refutation}.
\end{definition}

The notions of {\em degenerate pool resolution} proof and
{\em pool w-resolution} proof are
defined similarly, but allowing degenerate resolution or w-resolution
inferences, respectively.
The two papers
\cite{VanGelder:PoolResolution,HBPvG:clauselearn} defined pool
resolution to be the degenerate pool resolution system, so our notion
of pool resolution is more restrictive than theirs.  Our definition
is equivalent to the one in~\cite{Buss:poolhard}, however.  It is
also equivalent to the system regRTL defined in~\cite{BHJ:ResTreeLearning}.
Pool w-resolution is the same as
the system regWRTL of~\cite{BHJ:ResTreeLearning}.

A ``lemma'' in clause~(a) of the above definition
is called an {\em input lemma} if it is derived by {\em input}
subderivation, namely by a subderivation
in which each inference has at least one
hypothesis which is a member of~$F$ or is a lemma.
The notion of input lemma was first introduced
by~\cite{BHJ:ResTreeLearning}.  In their terminology,
a pool resolution proof which uses only input lemmas,
is called a regRTI proof.  Likewise a regWRTL
proof that uses only input lemmas is called a
regWRTI proof.

To understand the nomenclature; ``reg''
stands for ``regular'', ``W'' for ``w-resolution'',
``RT'' for ``resolution tree'', ``L'' for lemma,
and ``I'' for ``input lemma''.

Next we define various graph tautologies, sometimes also
called ``ordering principles''.  They will all
use a size parameter~$n>1$, and variables $x_{i,j}$ with $i,j\in[n]$
and $i\not= j$, where $[n] = \{0,1,2,\ldots,n{-}1\}$.   A variable~$x_{i,j}$
will intuitively represent the condition that $i\prec j$ with $\prec$
intended to be a total, linear order.  We will thus
always adopt the simplifying
convention that $x_{i,j}$ and~$\overline x_{j,i}$ are
the identical literal,
i.e., only the variables~$x_{i,j}$ for $i<j$ actually exist, and
$x_{j,i}$ for $j<i$ is just a notation for~$\overline x_{i,j}$, and
$\overline x_{j,i}$ stands for~$x_{i,j}$.
This identification makes no essential difference to
the complexity of proofs of the tautologies,
but it reduces the number of literals and clauses,
and simplifies the definitions.  In particular, it means there
are no axioms for the antisymmetry or totality of~$\prec$.

The following principle is based on the tautologies defined by
Krishnamurthy \cite{Krishnamurthy:trickyformulas}.
These tautologies, or similar ones,
have also been studied by \cite{Stalmarck:trickyformulas,%
BonetGalesi:ResAndPolyCalculus,%
AJPU:regularresolution,%
BeckmannBuss:dLK,%
SBI:SwitchingkDnf,%
VanGelder:InputCoverNumber,%
Johannsen:widthlearning}.

\begin{definition}
Let $n>1$.  Then $\GT_n$ is the following set of
clauses involving the variables $x_{i,j}$, for $i,j\in [n]$ with $i\not= j$.
\begin{enumerate}
\item[($\alpha_\emptyset$)]
The clauses $\bigvee_{j \not= i} x_{j,i}$, for each
value $i<n$.
\item[($\gamma_\emptyset$)] The {\em transitivity clauses} $T_{i,j,k}:=
\overline x_{i,j} \lor \overline x_{j,k} \lor \overline x_{k,i}$
for all distinct $i,j,k$ in $[n]$.
\end{enumerate}
\end{definition}

Note that the clauses $T_{i,j,k}$, $T_{j,k,i}$ and $T_{k,i,j}$ are identical.
For this reason Van Gelder \cite{VanGelder:PoolResolution}
uses the name ``no triangles'' (NT) for a similar principle.

The next definition is from~\cite{AJPU:regularresolution}, who
used the notation $\GT^\prime_n$.  They
used particular functions $r$ and~$s$ for their lower bound proof,
but since our upper bound proof does not depend on
the details of $r$ and~$s$
we leave them unspecified. We require that $r(i,j,k)\not=s(i,j,k)$ and that
the set $\{r(i,j,k),s(i,j,k)\}\not\subset\{i,j,k\}$.
In addition, w.l.o.g.,
$r(i,j,k)=r(j,k,i)=r(k,i,j)$, and similarly for~$s$.

\begin{definition}  Let $n\ge 1$, and let $r(i,j,k)$ and $s(i,j,k)$ be
functions mapping $[n]^3 \rightarrow [n]$ as above.  The {\em guarded
graph tautology} formula $\GGT_n$ consists of the following
clauses:
\begin{enumerate}
\item[($\alpha_\emptyset$)]
The clauses $\bigvee_{j \not= i} x_{j,i}$, for each
value $i<n$.
\item[($\gamma^\prime_\emptyset$)] The {\em guarded} transitivity clauses
$T_{i,j,k}\lor x_{r,s}$
and $T_{i,j,k}\lor \overline x_{r,s}$,
for all distinct $i,j,k$ in $[n]$, where $r=r(i,j,k)$ and $s=s(i,j,k)$.
\end{enumerate}
\end{definition}
Note that the $\GGT_n$ clauses depend on the functions $r$ and $s$;
this is suppressed in the notation.
Our main result for the guarded
graph tautologies is:
\begin{theorem}\label{PoolResGgtThm}
The guarded graph tautology formulas $\GGT_n$ have
polynomial size pool (regRTL) resolution refutations.
\end{theorem}

The proof of Theorem~\ref{PoolResGgtThm} will construct
pool refutations in the form of regular tree-like refutations
with lemmas.
A key part of this is
learning transitive closure clauses
that are derived using resolution
on the guarded transitivity clauses of~$\GGT_n$.
A slightly modified construction, that uses a
result from~\cite{BHJ:ResTreeLearning},
gives instead tree-like regular resolution
refutations with {\em input} lemmas.  This will establish
the following:
\begin{theorem}\label{regRtiGgtThm}
The guarded graph tautology formulas $\GGT_n$ have
polynomial size, tree-like regular resolution refutations with input lemmas
(regRTI refutations).
\end{theorem}
A consequence of Theorem~\ref{regRtiGgtThm} is that
the $\GGT_n$ clauses can be shown unsatisfiable by
non-greedy polynomial size DPLL searches using clause learning.
This follows via
Theorem~5.6 of~\cite{BHJ:ResTreeLearning},
since the refutations of~$\GGT_n$ are regRTI, and hence regWRTI, proofs
in the sense of~\cite{BHJ:ResTreeLearning}.

However, as shown by Theorem~\ref{DPLLpolySizeThm}
in Section~\ref{GreedySect}, we can improve
the constructions of Theorems \ref{PoolResGgtThm} and~\ref{regRtiGgtThm}
to show that
the $\GGT_n$ principles can be refuted also by
{\em greedy} and {\em unit-propagating} polynomial size DPLL searches with
clause learning.

\section{Guarded graph tautology refutations}\label{MainPfSect}

The following theorem is an important ingredient of our upper bound proof.

\begin{theorem}\label{GtProofsThm}
{\rm (St{\aa}lmarck~\cite{Stalmarck:trickyformulas};
Bonet-Galesi~\cite{BonetGalesi:ResAndPolyCalculus};
Van Gelder~\cite{VanGelder:InputCoverNumber})}
The sets $\GT_n$ have regular resolution
refutations~$P_n$ of polynomial size~$O(n^3)$.
\end{theorem}

We do not include a direct proof of Theorem~\ref{GtProofsThm}
here, which can be found in
\cite{Stalmarck:trickyformulas,BonetGalesi:ResAndPolyCalculus,VanGelder:InputCoverNumber}.
The present paper uses the proofs~$P_n$ as a ``black box'';
the only property needed is that the $P_n$'s are
regular and polynomial size.
Lemma~\ref{BpoDerivationLm} below
is a direct generalization to Theorem~\ref{GtProofsThm}; in fact,
when specialized to
the case of $\pi =\emptyset$, it is identical to Theorem~\ref{GtProofsThm}.

The refutations~$P_n$ can be modified to give refutations
of $\GGT_n$ by first deriving each transitive clause~$T_{i,j,k}$
from the two guarded transitivity clauses of~$(\gamma^\prime_\emptyset)$.
This however destroys the regularity property, and
in fact no polynomial size regular refutations exist
for $\GGT_n$~\cite{AJPU:regularresolution}.

As usual, a {\em partial order} on $[n]$ is an
antisymmetric, transitive relation binary relation
on~$[n]$.  We shall be mostly interested in ``partial
specifications'' of partial orders: partial
specifications are not
required to be transitive.

\begin{definition} A {\em partial specification},~$\tau$, of a
partial order is a set of ordered pairs $\tau\subseteq[n]\times[n]$
which are consistent with some (partial) order.
The minimal partial order containing~$\tau$ is the
transitive closure of~$\tau$.  We
write $i\prec_\tau j$ to denote $\langle i,j\rangle \in \tau$,
and write $i\prec_\tau^* j$ to denote that $\langle i,j\rangle$ is in the
transitive closure of~$\tau$.

The {\em $\tau$-minimal} elements are the $i$'s such that
$j\prec_\tau i$ does not hold for any~$j$.
\end{definition}

We will be primarily interested in particular
kinds of partial orders, called ``bipartite'' partial orders,
that can be associated with partial orders.  A bipartite partial order is a
partial order that does not have any chain of inequalities
$x\prec y\prec z$.

\begin{definition}
A {\em bipartite partial order} is a binary relation~$\pi$
on $[n]$
such that the domain and range of~$\pi$
do not intersect.
The set of $\pi$-minimal elements is denoted $M_\pi$.
\end{definition}
The righthand side of Figure~\ref{BiPartiteExampleFig} shows an example.
The bipartiteness of~$\pi$ arises from the fact that $M_\pi$ and
$[n]\setminus M_\pi$ partition $[n]$ into two sets.
Note that if $i \prec_\pi j$, then $i\in M_\pi$ and $j\notin M_\pi$.
In addition, $M_\pi$~contains the isolated points of~$\pi$.
\begin{definition}
Let $\tau$ be a specification of a partial order.
The bipartite partial order~$\pi$ that is {\em associated with}
$\tau$ is defined by letting $i\prec_\pi j$ hold for precisely
those $i$ and~$j$ such that $i$ is $\tau$-minimal
and $ i \prec^*_\tau j$.
\end{definition}
It is easy to check that the $\pi$ associated with~$\tau$
is in fact a bipartite partial order.
The intuition is that $\pi$~retains only the information about
whether $i\prec^*_\tau j$ for minimal elements~$i$,
and forgets the ordering that $\tau$ imposes on
non-minimal elements. Figure~\ref{BiPartiteExampleFig} shows an example of
how to obtain a bipartite partial order from a partial specification.

\begin{figure}[t]
\begin{center}
\psset{unit=10mm,arrowscale=1.4 1.2}     
\begin{pspicture}(0,0)(4.7,2)
\pscircle*(0,0){2pt}
\uput[0](0,0){$1$}
\pscircle*(1,0){2pt}
\uput[0](1,0){$2$}
\pscircle*(2,0){2pt}
\uput[0](2,0){$3$}
\pscircle*(3,0){2pt}
\uput[0](3,0){$4$}
\pscircle*(4,0){2pt}
\uput[0](4,0){$5$}
\pscircle*(1.5,1){2pt}
\uput[0](1.5,1){$6$}
\pscircle*(2.5,1){2pt}
\uput[0](2.5,1){$7$}
\pscircle*(3.5,1){2pt}
\uput[0](3.5,1){$8$}
\pscircle*(4.5,1){2pt}
\uput[0](4.5,1){$9$}
\pscircle*(2,2){2pt}
\uput[0](2,2){$10$}
\pscircle*(4,2){2pt}
\uput[0](4,2){$11$}
\psline{->}(2,0)(1.5,1)
\psline{->}(1.5,1)(2,2)
\psline{->}(2.5,1)(2,2)
\psline{->}(3,0)(2.5,1)
\psline{->}(3,0)(3.5,1)
\psline{->}(4,0)(3.5,1)
\psline{->}(4,0)(4.5,1)
\psline{->}(4.5,1)(4,2)
\end{pspicture}
\hfill
\raise 10mm \hbox{$\Rightarrow$}
\hfill
\raise 5mm \hbox{
\begin{pspicture}(-1.5,0)(4.5,1)
\pscircle*(0,0){2pt}
\uput[0](0,0){$1$}
\pscircle*(1,0){2pt}
\uput[0](1,0){$2$}
\pscircle*(2,0){2pt}
\uput[0](2,0){$3$}
\pscircle*(3,0){2pt}
\uput[0](3,0){$4$}
\pscircle*(4,0){2pt}
\uput[0](4,0){$5$}
\pscircle*(1.5,1){2pt}
\uput[45](1.5,1){$6$}
\pscircle*(3,1){2pt}
\uput[90](3,1){$7$}
\pscircle*(3.5,1){2pt}
\uput[90](3.5,1){$8$}
\pscircle*(4,1){2pt}
\uput[90](4,1){$9$}
\pscircle*(2.5,1){2pt}
\uput[90](2.5,1){$10$}
\pscircle*(4.5,1){2pt}
\uput[90](4.5,1){$11$}
\rput(-1,1){$[n]-M_\pi$:}
\rput(-1,0){$M_\pi$:}
\psline{->}(2,0)(1.5,1)
\psline{->}(2,0)(2.5,1)
\psline{->}(3,0)(3,1)
\psline{->}(3,0)(2.5,1)
\psline{->}(3,0,1)(3,1)
\psline{->}(3,0)(3.5,1)
\psline{->}(4,0)(3.5,1)
\psline{->}(4,0)(4,1)
\psline{->}(4,0)(4.5,1)
\end{pspicture}
}
\end{center}
\caption{Example of a partial specification of a partial order (left)
and the associated bipartite partial order (right).}
\label{BiPartiteExampleFig}
\end{figure}
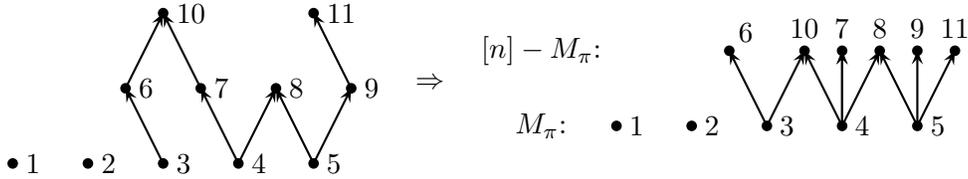

We define the graph tautology~$\GT_{\pi,n}$ relative to~$\pi$
as follows.
\begin{definition} Let $\pi$ be a bipartite partial order on~$[n]$.
Then $\GT_{\pi,n}$ is the set of clauses containing:
\begin{enumerate}
\item[($\alpha$)] The clauses $\bigvee_{j \not= i} x_{j,i}$, for each
value $i\in M_\pi$.
\item[($\beta$)] The transitivity clauses $T_{i,j,k}:=
\overline x_{i,j} \lor \overline x_{j,k} \lor \overline x_{k,i}$
for all distinct $i,j,k$ in $M_\pi$. (Vertices $i,j,k^\prime$
in Figure~\ref{BipartiteFig} show an example.)
\item[($\gamma$)] The transitivity clauses $T_{i,j,k}$
for all distinct $i,j,k$ such that $i,j\in M_\pi$ and
$i\not\prec_\pi k$ and $j\prec_\pi k$.
(As shown in
Figure~\ref{BipartiteFig}.)
\end{enumerate}
\end{definition}

The set $\GT_{\pi,n}$ is satisfiable if $\pi$ is nonempty.
As an example, there is the assignment that sets $x_{j,i}$ true
for some fixed $j\notin M_\pi$ and every $i\in M_\pi$, and sets
all other variables false.
However, if $\pi$
is applied as a restriction,
then $\GT_{\pi,n}$ becomes unsatisfiable.
That is to say, there is no assignment which
satisfies $\GT_{\pi,n}$ and is consistent with~$\pi$.
This fact
is proved by the regular derivation~$P_\pi$ described in the next
lemma.

\begin{definition} For $\pi$ a bipartite partial order, the clause
$\falseBPOpi$ is defined by
\[
\falseBPOpi ~:=~ \{ \overline x_{i,j} : i\prec_\pi j \},
\]
\end{definition}

\begin{lemma}\label{BpoDerivationLm}
Let $\pi$ be a bipartite partial order on~$[n]$.
Then there
is a regular derivation~$P_\pi$ of $\falseBPOpi$
from the set $\GT_{\pi,n}$.

The only variables resolved on in~$P_\pi$ are the following:
the variables $x_{i,j}$ such that $i,j\in M_\pi$,
and the variables $x_{i,k}$ such that $k\notin M_\pi$,
$i\in M_\pi$, and
$i\not\prec_\pi k$.
\end{lemma}

Lemma~\ref{BpoDerivationLm} implies
that if $\pi$ is the bipartite partial order
associated with a partial specification~$\tau$ of
a partial order, then the derivation~$P_\pi$ does
not resolve on any literal whose value is set
by~$\tau$.  This is proved by noting that if $i\prec_\tau j$,
then $j \notin M_\pi$.

Note that if $\pi$ is empty,
$M_\pi=[n]$ and there are no clauses
of type~($\gamma$).
In this case, $\GT_{\pi,n}$ is identical to~$\GT_n$,
and $P_\pi$ is the same as the refutation of~$\GT_n$
of Theorem~\ref{GtProofsThm}.

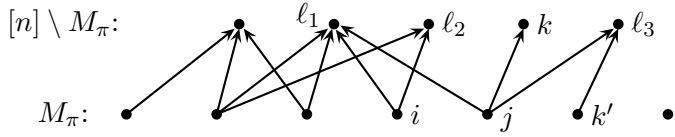
\begin{figure}[t]
\begin{center}
\psset{unit=0.06cm}     
\begin{pspicture}(-30,0)(140,20)
\pscircle*(0,0){2pt}
\pscircle*(20,0){2pt}
\pscircle*(40,0){2pt}
\pscircle*(60,0){2pt}
\pscircle*(80,0){2pt}
\pscircle*(100,0){2pt}
\pscircle*(120,0){2pt}
\rput(-14,0){$M_\pi$:}
\rput(-14,20){$[n]\setminus M_\pi$:}
\pscircle*(25,20){2pt}
\pscircle*(46,20){2pt}
\pscircle*(67,20){2pt}
\pscircle*(88,20){2pt}
\pscircle*(109,20){2pt}
\psline[arrowscale=1.4 1.2]{->}(0,0)(24,19)
\psline[arrowscale=1.4 1.2]{->}(20,0)(25,19)
\psline[arrowscale=1.4 1.2]{->}(40,0)(26,19)
\psline[arrowscale=1.4 1.2]{->}(20,0)(44.7,19)
\psline[arrowscale=1.4 1.2]{->}(40,0)(45.6,19)
\psline[arrowscale=1.4 1.2]{->}(60,0)(46.4,18.8)
\psline[arrowscale=1.4 1.2]{->}(80,0)(47.4,19.7)
\psline[arrowscale=1.4 1.2]{->}(20,0)(66.0,19)
\psline[arrowscale=1.4 1.2]{->}(60,0)(67,19)
\psline[arrowscale=1.4 1.2]{->}(80,0)(88,19)
\psline[arrowscale=1.4 1.2]{->}(80,0)(108,19)
\psline[arrowscale=1.4 1.2]{->}(100,0)(109,19)
\uput[0](60,0){$i$}
\uput[-10](80,0){$j$}
\uput[0](88,20){$k$}
\uput[0](100,0){$k^\prime$}
\uput[160](46,20){$\ell_1$}
\uput[0](67,20){$\ell_2$}
\uput[0](109,20){$\ell_3$}
\end{pspicture}
\end{center}
\caption{A bipartite partial order~$\pi$ is pictured, with the
ordered pairs of~$\pi$ shown as directed
edges.  (For instance, $j \prec_\pi k$ holds.)
The set $M_\pi$ is the set of minimal vertices.
The nodes $i,j,k$ shown are an example of nodes used
for a transitivity axiom
$\overline x_{i,j} \lor \overline x_{j,k} \lor \overline x_{k,i}$
of type~($\gamma$).  The nodes
$i,j,k^\prime$ are an example of the nodes
for a transitivity axiom of
type~($\beta$).}
\label{BipartiteFig}
\end{figure}

\begin{proof}
By renumbering the vertices, we can assume w.l.o.g.\ that
$M_\pi = \{0,\ldots,m{-}1\}$.
For each $k\ge m$, there is
at least one value of~$j$ such that $j \prec_\pi k$:
let $J_k$ be an arbitrary such value~$j$.  Note $J_k<m$.

Fix $i\in M_\pi$; that is, $i<m$.
Recall that the clause of type~($\alpha$) in $\GT_{\pi,n}$ for~$i$ is
$\bigvee_{j \not= i} x_{j,i}$.  We resolve this
clause successively, for each $k\ge m$ such that $i\not\prec_\pi k$,
against the clauses~$T_{i,J_k,k}$ of type~($\gamma$)
\[
     \overline x_{i,J_k}
\lor
     \overline x_{J_k,k}
\lor
     \overline x_{k,i}
\]
using resolution variables $x_{k,i}$.
(Note that $J_k\not= i$ since $i\not\prec_\pi k$.)
This yields a clause~$T^\prime_{i,m}$:
\[
\bigvee_{k\ge m \atop i \not\prec_\pi k} \overline x_{i,J_k}
~\lor~
\bigvee_{k\ge m \atop i \not\prec_\pi k} {\overline x_{J_k, k}}
~\lor~
\bigvee_{k\ge m \atop i \prec_\pi k}{x_{k,i}}
~\lor~
\bigvee_{k<m \atop k\not = i} x_{k,i}.
\]
The first two disjuncts shown above for~$T^\prime_{i,m}$
come from the side literals of the clauses~$T_{i,J_k,k}$;
the last two disjuncts come
from the literals in $\bigvee_{j \not= i} x_{j,i}$
which were not resolved on.
Since a literal $\overline x_{i,J_k}$ is the
same literal as $x_{J_k,i}$ and since $J_k<m$,
the literals in the first disjunct are
also contained in the fourth disjunct.
Thus, eliminating duplicate literals, $T^\prime_{i,m}$ is
equal to the clause
\[
\bigvee_{k\ge m \atop i \not\prec_\pi k} {\overline x_{J_k, k}}
~\lor~
\bigvee_{k\ge m \atop i \prec_\pi k}{x_{k,i}}
~\lor~
\bigvee_{k<m \atop k\not = i} x_{k,i}.
\]

Repeating this process, we obtain
derivations of the clauses~$T^\prime_{i,m}$ for all $i<m$.
The final disjuncts of these clauses,
$\bigvee_{i\not= k<m} x_{k,i}$,
are the same as the ($\alpha_\emptyset$) clauses in $\GT_m$.
Thus, the clauses~$T^\prime_{i,m}$ give
all ($\alpha_\emptyset$) clauses of~$\GT_m$, but with literals
$\overline x_{J_k, k}$ and~$x_{k,i}$ added in as side literals.
Moreover, the
clauses of type~($\beta$) in $\GT_{\pi,n}$
are exactly the transitivity clauses
of $\GT_m$.   All these clauses can be combined
exactly as in the refutation of~$\GT_m$
described in Theorem~\ref{GtProofsThm}, but carrying along
extra side literals $\overline x_{J_k,k}$ and~$x_{k,i}$,
or equivalently carrying along
literals~$\overline x_{J_k,k}$ for $J_k \prec_\pi k$,
and~$\overline x_{i,k}$ for $i \prec_\pi k$.
Since the refutation of~$\GT_m$ uses all of its transitivity
clauses and since each $\overline x_{J_k,k}$ literal is also
one of the $\overline x_{i,k}$'s,
this yields a resolution derivation~$P_\pi$
of the clause
\[
\{  \overline x_{i,k} : \hbox{$i\prec_\pi k$} \}.
\]
This is the clause $\falseBPOpi$ as desired.

Finally, we observe that $P_\pi$ is regular.  To show this,
note that the first parts of $P_\pi$ deriving the clauses~$T^\prime_{i,m}$
are regular by construction,
and they use resolution only on variables~$x_{k,i}$
with $k\ge m$, $i<m$, and $i\not\prec_\pi k$.
The remaining part of~$P_\pi$ is also regular by Theorem~\ref{GtProofsThm},
and uses resolution only on variables $x_{i,j}$ with $i,j\le m$.
\hfill $\qed$
\end{proof}

\begin{proof} of Theorem~\ref{PoolResGgtThm}.
We will show how to construct a series of ``LR partial refutations'',
denoted $R_0$, $R_1$, $R_2, \ldots$; this process
eventually terminates
with a pool (regRTL) resolution refutation of~$\GGT_n$.
The terminology ``LR partial'' indicates that
the refutation is being constructed in left-to-right order, with
the left part of the refutation properly formed, but with
many of the remaining leaves being labeled with bipartite partial orders
instead of with valid learned clauses or initial clauses from~$\GGT_n$.
We first describe the construction of the pool refutation, and
leave the size analysis to the end.

An LR partial refutation~$R$ is a tree with nodes
labeled with clauses that form a correct pool resolution proof,
except possibly at the leaves.
Furthermore, it must satisfy the following conditions.

\begin{description}
\item[\rm a.]
$R$~is a tree. The root is labeled with the
empty clause. Each non-leaf node in~$R$ has a left child and right
child; the clause labeling the node is derived by resolution from the clauses on its two children.
\item[\rm b.] For each clause~$C$ occurring in~$R$, the
clause $C^+$ and the set of ordered
pairs $\tau(C)$ are defined by
\begin{eqnarray*}
C^+ & := & \{\overline x_{i,j} : \hbox{$\overline x_{i,j}$
occurs in some clause on the branch} \\
& & \quad \quad \quad \quad \quad
\hbox{from the root node of~$R$ up to and including~$C$}\},
\end{eqnarray*}
and $\tau(C) = \{ \langle i,j \rangle : \overline x_{i,j}\in C^+\}$.
Note that $C\subseteq C^+$ holds by definition.
In many cases,
$\tau(C)$ will be a partial specification of a partial order,
but this is not always true.  For instance, if $C$ is
a transitivity axiom, then $\tau(C)$ has a 3-cycle
and is not consistent
as a specification of a partial order.

\item[\rm d.] Leaves are either ``finished'' or ``unfinished''.
Each finished leaf~$L$ is
labeled with either a clause
from $\GGT_n$ or a clause that occurs to the left of~$L$
in the postorder traversal of~$R$.
\item[\rm e.] For an unfinished leaf labeled with clause~$C$,
the set $\tau(C)$ is
a partial specification of a partial order.
Furthermore, letting
$\pi$ be the bipartite
partial order associated with $\tau(C)$,
the clause $C$ is equal to~$\falseBPOpi$.
\end{description}

Property e.\ is particularly crucial and is novel to
our construction.  As shown below, each
unfinished leaf, labeled with a clause $C = \falseBPOpi$,
will be replaced by a derivation~$S$.  The derivation~$S$ often will
be based on~$P_\pi$, and thus might be expected to end with
exactly the clause~$C$; however, some of the resolution inferences
needed for $P_\pi$ might be disallowed by the regularity property of pool
resolution proofs.
This can mean that $S$ will instead be
a derivation of a clause~$C^\prime$ such that
$C\subseteq C^\prime \subseteq C^+$.
The condition $C^\prime\subseteq C^+$ is
required because any literal $x \in C^\prime \setminus C$
will be handled by modifying the refutation~$R$
by propagating $x$ downward in~$R$
until reaching a clause that already contains~$x$.
The condition $C^\prime\subseteq C^+$ ensures that such
a clause exists.
The fact that $C^\prime\supseteq C$
will mean that enough literals are present
for the derivation to
use only (non-degenerate) resolution
inferences --- by virtue of the fact that our
constructions will pick~$C$ so that it contains
the literals that must be present for use as
resolution literals.  The extra literals
in $C^\prime \setminus C$ will be handled by propagating them down
the proof to where they are resolved on.

The construction begins by letting $R_0$ be the ``empty'' refutation,
containing just the empty clause.  Of course, this clause is
an unfinished leaf, and $\tau(\emptyset) = \emptyset$.  Thus
$R_0$ is a valid LR partial refutation.

For the induction step, $R_i$ has been constructed already.
Let $C$ be the leftmost unfinished clause in~$R_i$.
$R_{i+1}$ will be formed by replacing~$C$
by a refutation~$S$ of some clause
$C^\prime$ such that $C\subseteq C^\prime \subseteq C^+$.

We need to describe the (LR partial) refutation~$S$.  Let $\pi$ be the
bipartite partial order associated with $\tau(C)$, and consider
the derivation~$P_\pi$ from Lemma~\ref{BpoDerivationLm}.
Since $C$ is $\falseBPOpi$ by condition~e., the final line of~$P_\pi$
is the clause~$C$.  The intuition is that we would like
to let $S$ be~$P_\pi$.  The first difficulty with this is that
$P_\pi$ is dag-like, and the $LR$-refutation is intended to be
tree-like,  This difficulty, however, can be circumvented by just
expanding $P_\pi$, which is regular,
into a tree-like regular derivation with lemmas by the simple expedient of
using a depth-first traversal of~$P_\pi$.
The second, and more serious, difficulty is that
$P_\pi$ is
a derivation from~$\GT_n$, not~$\GGT_n$.  Namely,
the derivation~$P_\pi$
uses the transitivity clauses of~$\GT_n$ as initial clauses
instead of the guarded
transitivity clauses of~$\GGT_n$.
The transitivity clauses $T_{i,j,k} :=
\overline x_{i,j}\lor \overline x_{j,k} \lor \overline x_{k,i}$
in~$P_\pi$ are handled one at a time as described below.
We will use four separate constructions:
in case~\caseiti{}, no change to~$P_\pi$ is required;
cases \caseitii{} and~\caseitiii{} require small changes;
and in the fourth case, the subproof~$P_\pi$
is abandoned in favor of ``learning'' the transitivity clause.

By the remark made after
Lemma~\ref{BpoDerivationLm},
no literal in~$C^+$ is used as a resolution
literal in~$P_\pi$.

\begin{description}
\item[\caseiti] If a transitivity clause~$T_{i,j,k}$
of~$P_\pi$ already appears earlier in~$R_i$ (that is, to the left
of~$C$), then it is already {\em learned}, and can be used freely
in~$P_\pi$.
\end{description}
In the remaining cases \caseitii{}-\caseitiv{},
the transitivity clause~$T_{i,j,k}$ is not yet learned.
Let the guard variable for~$T_{i,j,k}$
be $x_{r,s}$, so $r=r(i,j,k)$ and $s=s(i,j,k)$.
\begin{description}
\item[\caseitii] Suppose case~\caseiti{} does not
apply and
that the guard variable~$x_{r,s}$ or its negation $\overline x_{r,s}$
is a member of~$C^+$.
The guard variable thus is used as a resolution variable somewhere
along the branch from the root to clause~$C$.
Then, as mentioned above,
Lemma~\ref{BpoDerivationLm} implies that $x_{r,s}$ is not resolved on
in~$P_\pi$.
Therefore,
we can add the literal $x_{r,s}$ or~$\overline x_{r,s}$ (respectively)
to the clause $T_{i,j,k}$ and to every clause on any path below~$T_{i,j,k}$
until reaching a clause that already contains that literal.
This replaces $T_{i,j,k}$ with one of the initial clauses
$T_{i,j,k}\lor x_{r,s}$ or
$T_{i,j,k}\lor \overline x_{r,s}$ of~$\GGT_n$.
By construction, it preserves the validity of the resolution
inferences of~$R_i$ as well as the regularity
property.
Note this adds the literal
$x_{r,s}$ or $\overline x_{r,s}$ to the final
clause~$C^\prime$ of the modified~$P_\pi$.
This maintains the property that
$C\subseteq C^\prime \subseteq C^+$.
\item[\caseitiii] Suppose case~\caseiti{} does not apply and that
$x_{r,s}$ is not used as a
resolution variable anywhere below~$T_{i,j,k}$ in~$P_\pi$ and
is not a member of $C^+$.  In this case, $P_\pi$
is modified so as to derive the clause~$T_{i,j,k}$
from the two $\GGT_n$ clauses
$T_{i,j,k}\lor x_{r,s}$ and
$T_{i,j,k}\lor \overline x_{r,s}$ by resolving on~$x_{r,s}$.
This maintains the regularity of the derivation.  It also means that
henceforth $T_{i,j,k}$ will be learned.
\end{description}
If all of the transitivity clauses in~$P_\pi$ can be handled by
cases \caseiti{}-\caseitiii{},
then we use $P_\pi$ to define~$R_{i+1}$.  Namely,
let $P_\pi^\prime$ be the derivation~$P_\pi$
as modified by the applications of cases \caseitii{}
and~\caseitiii{}.
The derivation~$P_\pi^\prime$ is regular and dag-like, so we
can recast it as a tree-like derivation~$S$ with lemmas, by using
a depth-first traversal of~$P^\prime_\pi$.
The size of~$S$ is linear in the size of~$P^\prime_\pi$, since only
input lemmas need to be repeated.  The final line of~$S$
is the clause~$C^\prime$, namely $C$ plus the literals
introduced by case~\caseitii{}.
The derivation $R_{i+1}$ is
formed from~$R_i$ by replacing the clause~$C$ with the derivation~$S$
of~$C^\prime$, and then propagating each new literal $x\in C^\prime\setminus C$
down towards the root of~$R_i$, adding~$x$ to each clause below~$S$
until reaching a clause that already contains~$x$.
The derivation~$S$ contains no unfinished leaf, so $R_{i+1}$ contains
one fewer unfinished leaves than~$R_i$.

On the other hand, if even one transitivity axiom~$T_{i,j,k}$ in~$P_\pi$ is
not covered by the above three cases, then case~\caseitiv{} must be
used instead.
This introduces
a completely different construction to form~$S$:
\begin{description}
\item[\caseitiv{}] Let $T_{i,j,k}$ be any transitivity axiom in~$P_\pi$
that is not covered by cases \caseiti{}-\caseitiii{}.
In this case, the guard variable~$x_{r,s}$ is used as a resolution
variable in~$P_\pi$ somewhere below~$T_{i,j,k}$; in general,
this means we cannot use resolution on~$x_{r,s}$ to derive $T_{i,j,k}$
while maintaining the desired pool property.  Hence, $P_\pi$~is no longer used,
and
we instead will form $S$ with a short left-branching path
that
``learns'' $T_{i,j,k}$.
This will generate two or three new unfinished leaf nodes.  Since unfinished
leaf nodes in a LR partial derivation must be labeled with clauses
from bipartite partial orders, it is also necessary to attach short
derivations to these
unfinished leaf nodes to make the unfinished leaf clauses
of~$S$ correspond correctly to
bipartite partial orders.
These unfinished leaf nodes are then kept in~$R_{i+1}$
to be handled at later stages.

There are separate constructions depending on whether
$T_{i,j,k}$ is a clause of type ($\beta$) or~($\gamma$);
details are given below.
\end{description}

First suppose $T_{i,j,k}$ is of type~($\gamma$), and thus
$\overline x_{j,k}$ appears in~$C$.
(Refer to Figure~\ref{BipartiteFig}.)
Let $x_{r,s}$ be the
guard variable for the transitivity axiom~$T_{i,j,k}$.
The derivation~$S$ will have the form
\label{gammaCaseEq}
\begin{prooftree}
\AxiomC{$\overline x_{i,j},\overline x_{j,k}, \overline x_{k,i}, x_{r,s}$}
\AxiomC{$\overline x_{i,j},\overline x_{j,k}, \overline x_{k,i}, \overline x_{r,s}$}
\BinaryInfC{$\overline x_{i,j},\overline x_{j,k}, \overline x_{k,i}$}
\AxiomC{\raisebox{7pt}{$S_1$}$\proofdots$}
\kernHyps{-1ex}
\noLine
\UnaryInfC{$\overline x_{i,j},\overline x_{i,k},\overline \pi_{-[jk;jR(i)]}$}
\BinaryInfC{$\overline x_{i,j},\overline x_{j,k},\overline \pi_{-[jk;jR(i)]}$}
\AxiomC{\raisebox{7pt}{$S_2$}$\proofdots$}
\kernHyps{-1ex}
\noLine
\UnaryInfC{$\overline x_{j,i},\overline x_{j,k},\overline \pi_{-[jk;iR(j)]}$}
\BinaryInfC{$\overline x_{j,k}, \overline \pi_{-[jk]}$}
\end{prooftree}
The notation $\overline \pi_{-[jk]}$ denotes the
disjunction of the negations of the literals in
$\pi$ omitting the literal~$\overline x_{j,k}$.
We write ``$iR(j)$'' to indicate literals $x_{i,\ell}$
such that $j\prec_\pi \ell$.  (The ``$R(j)$''
means ``range of~$j$''.)  Thus $\overline \pi_{-[jk;iR(j)]}$
denotes the clause containing the negations of the literals in~$\pi$,
omitting $\overline x_{j,k}$ and any literals $\overline x_{i,\ell}$ such
that $j\prec_\pi \ell$.  The clause $\overline \pi_{-[jk;jR(i)]}$
is defined similarly.

The upper leftmost inference of~$S$ is a resolution inference on
the variable~$x_{r,s}$.  Since $T_{i,j,k}$ is not covered by
either case \caseiti{} or~\caseitii{}, the variable~$x_{r,s}$
is not in~$C^+$.  Thus,
this use of~$x_{r,s}$ as a resolution variable does not violate
regularity.  Furthermore, since $T_{i,j,k}$ is of type~($\gamma$),
we have
$i{\not\prec_{\tau(C)}} j$,
$j{\not\prec_{\tau(C)}} i$,
$i{\not\prec_{\tau(C)}} k$, and
$k{\not\prec_{\tau(C)}} i$.
Thus
the literals $x_{i,j}$ and~$x_{i,k}$ are not in~$C^+$, so
they also can be resolved on without violating regularity.

Let $C_1$ and~$C_2$ be the final clauses of $S_1$ and~$S_2$, and
let $C_1^-$ be the clause below~$C_1$ and above~$C$.  The set
$\tau(C_2)$
is obtained by adding $\langle j,i \rangle$ to $\tau(C)$,
and similarly $\tau(C_1^-)$ is $\tau(C)$ plus~$\langle i,j \rangle$.
Since $T_{i,j,k}$ is type~($\gamma$), we have $i,j\in M_\pi$.
Therefore,
since $\tau(C)$ is a partial specification of a partial order,
$\tau(C_2)$ and $\tau(C_1^-)$ are also both partial specifications
of partial orders.
Let $\pi_2$ and~$\pi_1$ be the bipartite orders
associated with these two partial specifications (respectively).
We will form the subproof~$S_1$
so that it contains the clause $\falseBPOpisub 1$ as its only unfinished clause.
This will require adding inferences in~$S_1$ which
add and remove the appropriate literals.  The first step of this type
already occurs in going up from $C_1^-$ to~$C_1$ since this
has removed $\overline x_{j,k}$ and added $\overline x_{i,k}$,
reflecting the fact that $j$ is not $\pi_1$-minimal
and thus $x_{i,k}\in\pi_1$ but $x_{j,k}\notin \pi_1$.
Similarly, we will form~$S_2$ so that its only unfinished clause
is $\falseBPOpisub 2$.

We first describe the
subproof~$S_2$ of~$S$.  The situation is pictured in Figure~\ref{S2Fig},
which shows an extract from Figure~\ref{BipartiteFig}: the edges
shown in part~(a) of the figure correspond to the literals present in
the final line~$C_2$ of~$S_2$.  In particular, recall that the
literals $\overline x_{i,\ell}$ such that $j\prec_\pi \ell$ are
omitted from the last line of~$S_2$.  (Correspondingly, the edge
from $i$ to~$\ell_1$ is omitted from Figure~\ref{S2Fig}.)   The last
line~$C_2$ of~$S_2$ may not correspond to a bipartite partial order
as it may not partition $[n]$ into minimal and non-minimal elements;
thus, the last line of~$S_2$
may not qualify to be an unfinished node of~$R_{i+1}$.
(An example of this in Figure~\ref{S2Fig}(a) is
that $j\prec_{\tau(C_2)} i \prec_{\tau(C_2)}\ell_2$, corresponding to
$\overline x_{j,i}$ and~$\overline x_{i,\ell_2}$
being in the last line of~$S_2$.)
The bipartite partial
order~$\pi_2$ associated with~$\tau(C_2)$ is equal to
the bipartite partial order that agrees with~$\pi$ except that each
$i\prec_\pi \ell$ condition is replaced with
the condition $j\prec_{\pi_2} \ell$.  (This is represented in
Figure~\ref{S2Fig}(b) by the fact that
the edge from $i$ to~$\ell_2$ has been replaced by the edge
from $j$ to~$\ell_2$. Note that the vertex~$i$ is no longer a
minimal element of~$\pi_2$; that is, $i\notin M_{\pi_2}$.)
We wish to form~$S_2$ to be a regular derivation of
the clause $\overline x_{j,i},\overline\pi_{-[jk;iR(j)]}$ from
the clause $\falseBPOpisub 2$.

The subproof of~$S_2$ for replacing $\overline x_{i,\ell_2}$ in~$\overline \pi$
with $\overline x_{j,\ell_2}$ in~$\overline \pi_2$ is as follows, letting
$\overline \pi^*$ be $\overline\pi_{-[jk;iR(j);i\ell_2]}$.
\begin{equation}\label{S2FormEq}
\AxiomC{\raisebox{7pt}{$S^\prime_2$}$\proofdots$}
\kernHyps{-1ex}
\noLine
\UnaryInfC{$\overline x_{j,i},\overline x_{i,\ell_2},\overline x_{\ell_2,j}$}
\AxiomC{$\proofdots$ \raisebox{1ex}{\hbox to 0pt{rest of~$S_2$}}}
\noLine
\UnaryInfC{$\overline x_{j,k},\overline x_{j,\ell_2}, \overline x_{j,i}, \overline \pi^*$}
\BinaryInfC{$\overline x_{j,k},\overline x_{i,\ell_2},\overline x_{j,i},\overline \pi^*$}
\DisplayProof
\end{equation}
The part labeled ``rest of $S_2$'' will handle similarly the
other literals~$\ell$ such that
$i\prec_\pi \ell$ and $j\not\prec_\pi \ell$.
The final line of $S_2^\prime$ is the transitivity axiom~$T_{j,i,\ell_2}$.
This is a $\GT_n$ axiom, not a $\GGT_n$ axiom; however, it can be
handled by the methods of cases \caseiti{}-\caseitiii{}.
Namely, if $T_{j,i,\ell_2}$ has already been learned by appearing somewhere
to the left in~$R_i$, then $S^\prime_2$~is just this single clause.
Otherwise, let the guard variable for $T_{j,i,\ell_2}$ be $x_{r',s'}$.
If $x_{r',s'}$ is used as a resolution variable below~$T_{j,i,\ell_2}$, then
replace $T_{j,i,\ell_2}$ with
$T_{j,i,\ell_2}\lor x_{r',s'}$ or $T_{j,i\ell_2}\lor \overline x_{r',s'}$,
and propagate the $x_{r',s'}$ or $\overline x_{r',s'}$ to clauses down
the branch leading to~$T_{j,i,\ell_2}$ until reaching a clause that
already contains that literal.
Finally, if $x_{r',s'}$ has not been
used as a resolution variable in~$R_i$ below~$C$, then let
$S^\prime_2$~consist of a resolution inference deriving (and learning)
$T_{j,i,\ell_2}$ from
the clauses
$T_{j,i,\ell_2},x_{r',s'}$ and
$T_{j,i,\ell_2},\overline x_{r',s'}$.

To complete the construction of~$S_2$,
the inference (\ref{S2FormEq}) is repeated for each value of~$\ell$
such that $i\prec_\pi\ell$ and $j\not\prec_\pi\ell$.
The result is that $S_2$ has one unfinished leaf clause, and it
is labeled with the clause $\falseBPOpisub 2$.

\begin{figure}[t]
\psset{unit=0.06cm}     
\hfill
\begin{pspicture}(40,-10)(110,20)
\pscircle*(60,0){2pt}
\pscircle*(80,0){2pt}
\pscircle*(46,20){2pt}
\pscircle*(67,20){2pt}
\pscircle*(88,20){2pt}
\pscircle*(109,20){2pt}
\psline[arrowscale=1.4 1.2]{->}(80,0)(47.4,19.7)
\psline[arrowscale=1.4 1.2]{->}(60,0)(67,19)
\psline[arrowscale=1.4 1.2]{->}(80,0)(88,19)
\psline[arrowscale=1.4 1.2]{->}(80,0)(108,19)
\psline[arrowscale=1.4 1.2]{->}(80,0)(61,0)
\uput[180](60,0){$i$}
\uput[-10](80,0){$j$}
\uput[0](88,20){$k$}
\uput[160](46,20){$\ell_1$}
\uput[0](67,20){$\ell_2$}
\uput[0](109,20){$\ell_3$}
\rput(70,-10){(a) $\overline x_{j,k},\overline x_{i,\ell_2},\overline x_{j,i},\overline \pi^*$}
\end{pspicture}
\hfill
\begin{pspicture}(40,-10)(110,20)
\pscircle*(60,0){2pt}
\pscircle*(80,0){2pt}
\pscircle*(46,20){2pt}
\pscircle*(67,20){2pt}
\pscircle*(88,20){2pt}
\pscircle*(109,20){2pt}
\psline[arrowscale=1.4 1.2]{->}(80,0)(47.4,19.7)
\psline[arrowscale=1.4 1.2]{->}(80,0)(67,19)
\psline[arrowscale=1.4 1.2]{->}(80,0)(88,19)
\psline[arrowscale=1.4 1.2]{->}(80,0)(108,19)
\psline[arrowscale=1.4 1.2]{->}(80,0)(61,0)
\uput[180](60,0){$i$}
\uput[-10](80,0){$j$}
\uput[0](88,20){$k$}
\uput[160](46,20){$\ell_1$}
\uput[0](67,20){$\ell_2$}
\uput[0](109,20){$\ell_3$}
\rput(70,-10){(b) $\overline x_{j,k},\overline x_{i,\ell_2},\overline x_{j,i},\overline \pi^*$}
\end{pspicture}
\hfill
\caption{The partial orders for the fragment of~$S_2$ shown
in~(\ref{S2FormEq}).}
\label{S2Fig}
\end{figure}

We next describe the subproof~$S_1$ of~$S$.
The situation is shown in Figure~\ref{S1Fig}.  As in the formation
of~$S_2$, the final clause~$C_1$ in~$S_1$ may need to be modified in order
to correspond to the bipartite partial order~$\pi_1$ which
is associated with~$\tau(C_1)$.  First,
note that the literal $\overline x_{j,k}$ is already
replaced by~$\overline x_{i,k}$ in the final clause of~$S_1$.
The other change that is needed is that, for every $\ell$ such that
$j\prec_\pi \ell$ and $i\not\prec_\pi \ell$, we must
replace $\overline x_{j,\ell}$ with $\overline x_{i,\ell}$
since we have
$j\not\prec_{\pi_1} \ell$ and $i\prec_{\pi_1} \ell$.  Vertex~$\ell_3$
in Figure~\ref{S1Fig} is an example of a such a value~$\ell$.  The
ordering in the final clause of~$S_1$ is shown in part~(a), and the
desired ordered pairs of~$\pi_1$ are shown in part~(b).
Note that $j$ is no longer a minimal element in~$\pi_1$.

The replacement of $\overline x_{j,\ell_3}$
with $\overline x_{i,\ell_3}$
is effected by the following inference, letting
$\overline \pi^*$ now be
$\overline \pi_{-[jk;jR(i);j\ell_3]}$.
\begin{equation}\label{S1FormEq}
\AxiomC{\raisebox{7pt}{$S^\prime_1$}$\proofdots$}
\kernHyps{-1ex}
\noLine
\UnaryInfC{$\overline x_{i,j},\overline x_{j,\ell_3},\overline x_{\ell_3,i}$}
\AxiomC{$\proofdots$ \raisebox{1ex}{\hbox to 0pt{rest of~$S_1$}}}
\noLine
\UnaryInfC{$\overline x_{i,k},\overline x_{i,\ell_3}, \overline x_{i,j}, \overline \pi^*$}
\BinaryInfC{$\overline x_{i,k},\overline x_{j,\ell_3},\overline x_{i,j},\overline \pi^*$}
\DisplayProof
\end{equation}
The ``rest of $S_1$'' will handle similarly the
other literals~$\ell$ such that
$j\prec_\pi \ell$ and $i\not\prec_\pi \ell$.
Note that the final clause of~$S_1^\prime$ is the
transitivity axiom $T_{i,j,\ell_3}$.  The subproof $S_1^\prime$ is
formed in exactly the same way that $S_2^\prime$ was formed above. Namely,
depending on the status of the guard variable~$x_{r',s'}$
for~$T_{i,j,\ell_3}$, one of the following
is done:
\caseiti{}~the clause~$T_{i,j,\ell_3}$ is already learned
and can be used as is,
or \caseitii{}~one of $x_{r',s'}$ or~$\overline x_{r',s'}$ is
added to the clause and propagated down the proof,
or \caseitiii{}~the clause~$T_{i,j,\ell_3}$ is inferred
using resolution on~$x_{r',s'}$
and becomes learned.

To complete the construction of~$S_1$,
the inference (\ref{S1FormEq}) is repeated for each value of~$\ell$
such that $j\prec_\pi\ell$ and $i\not\prec_\pi\ell$.
The result is that $S_1$ has one unfinished leaf clause, and it corresponds to
the bipartite partial order~$\pi_1$.

\begin{figure}[t]
\psset{unit=0.06cm}     
\hfill
\begin{pspicture}(40,-10)(110,23)
\pscircle*(60,0){2pt}
\pscircle*(80,0){2pt}
\pscircle*(46,20){2pt}
\pscircle*(67,20){2pt}
\pscircle*(88,20){2pt}
\pscircle*(109,20){2pt}
\psline[arrowscale=1.4 1.2]{->}(60,0)(46.4,18.8)
\psline[arrowscale=1.4 1.2]{->}(60,0)(67,19)
\psline[arrowscale=1.4 1.2]{->}(60,0)(79,0)
\psline[arrowscale=1.4 1.2]{->}(60,0)(87.5,19)
\psline[arrowscale=1.4 1.2]{->}(80,0)(108,19)
\uput[180](60,0){$i$}
\uput[-10](80,0){$j$}
\uput[0](88,20){$k$}
\uput[160](46,20){$\ell_1$}
\uput[0](67,20){$\ell_2$}
\uput[0](109,20){$\ell_3$}
\rput(70,-10){(a) $\overline x_{i,k},\overline x_{j,\ell_3},\overline x_{i,j},\overline \pi^*$}
\end{pspicture}
\hfill
\begin{pspicture}(40,-10)(120,23)
\pscircle*(60,0){2pt}
\pscircle*(80,0){2pt}
\pscircle*(46,20){2pt}
\pscircle*(67,20){2pt}
\pscircle*(88,20){2pt}
\pscircle*(109,20){2pt}
\psline[arrowscale=1.4 1.2]{->}(60,0)(46.4,18.8)
\psline[arrowscale=1.4 1.2]{->}(60,0)(67,19)
\psline[arrowscale=1.4 1.2]{->}(60,0)(87.5,19)
\psline[arrowscale=1.4 1.2]{->}(60,0)(107.5,19)
\psline[arrowscale=1.4 1.2]{->}(60,0)(79,0)
\uput[180](60,0){$i$}
\uput[-10](80,0){$j$}
\uput[0](88,20){$k$}
\uput[160](46,20){$\ell_1$}
\uput[0](67,20){$\ell_2$}
\uput[0](109,20){$\ell_3$}
\rput(70,-10){(b) $\overline x_{i,k},\overline x_{i,\ell_3},\overline x_{i,j},\overline \pi^*$}
\end{pspicture}
\hfill
\caption{The partial orders for the fragment of~$S_1$ shown
in~(\ref{S1FormEq}).}
\label{S1Fig}
\end{figure}

That completes the construction of the subproof~$S$ for the subcase
of~\caseitiv{} where $T_{i,j,k}$ is of type~($\gamma$).  Now suppose
$T_{i,j,k}$ is of type~($\beta$).  (For instance, the values $i,j,k^\prime$
of Figure~\ref{BipartiteFig}.)
In this case the derivation~$S$ will have the form
\label{betaCaseEq}
\begin{prooftree}
\AxiomC{$T_{i,j,k}, x_{r,s}$}
\AxiomC{$T_{i,j,k}, \overline x_{r,s}$}
\BinaryInfC{$T_{i,j,k}$}
\AxiomC{\raisebox{7pt}{$S_3$}$\proofdots$}
\kernHyps{-1ex}
\noLine
\UnaryInfC{$\overline x_{i,j},\overline x_{i,k},\overline \pi_{-[jR(i),kR(i\cup j)]}$}
\BinaryInfC{$\overline x_{i,j},\overline x_{j,k},\overline \pi_{-[jR(i),kR(i\cup j)]}$}
\AxiomC{\raisebox{7pt}{$S_4$}$\proofdots$}
\kernHyps{-1ex}
\noLine
\UnaryInfC{$\overline x_{i,j},\overline x_{k,j},\overline \pi_{-[jR(i\cap k)]}$}
\BinaryInfC{$\overline x_{i,j},\overline \pi_{-[jR(i\cap k)]}$}
\AxiomC{\raisebox{7pt}{$S_5$}$\proofdots$}
\kernHyps{-1ex}
\noLine
\UnaryInfC{$\overline x_{j,i},\overline \pi_{-[iR(j)]}$}
\BinaryInfC{$\overline \pi$}
\end{prooftree}
where $x_{r,s}$ is the guard variable for~$T_{i,j,k}$.
We write $[\overline \pi_{-[jR(i\cap k)]}]$ to mean
the negations of literals in~$\pi$
omitting any literal $\overline x_{j,\ell}$ such that
both $i\prec_\pi \ell$ and $k\prec_\pi\ell$.
Similarly, $\overline \pi_{-[jR(i),kR(i\cup j)]}$ indicates
the negations of literals in~$\pi$,
omitting the literals $\overline x_{j,\ell}$
such that $i\prec_\pi \ell$ and the literals
$\overline x_{k,\ell}$ such that either
$i\prec_\pi \ell$ or $j\prec_\pi \ell$.

Note that
the resolution on~$x_{r,s}$ used to derive $T_{i,j,k}$ does
not violate regularity, since otherwise $T_{i,j,k}$ would
have been covered by case~\caseitii{}.  Likewise,
the resolutions on $x_{i,j}$, $x_{i,k}$ and $x_{j,k}$ do not
violate regularity since $T_{i,j,k}$ is
of type~($\beta$).


The subproof~$S_5$
is formed exactly like the subproof~$S_2$ above, with the exception
that now the literal $\overline x_{j,k}$ is not present.  Thus we omit
the description of~$S_5$.

We next describe the construction of the
subproof $S_4$.  Let $C_4$ be the final clause of~$S_4$;
it is easy to check that $\tau(C_4)$ is a partial specification
of a partial order.
As before, we must derive~$C_4$
from the clause $\falseBPOpisub 4$ where $\pi_4$ is the bipartite
partial order associated with the partial order~$\tau(C_4)$.
A typical situation is shown in Figure~\ref{S4Fig}.  As pictured
there, it is necessary to add the
literals $\overline x_{i,\ell}$ such that $j\prec_\pi \ell$ and
$i\not\prec_\pi \ell$, while removing $\overline x_{j,\ell}$; examples
of this are $\ell$ equal to $\ell_2$ and~$\ell_3$
in Figure~\ref{S4Fig}.
At the same time, we must add the
literals $\overline x_{k,\ell}$ such that $j\prec_\pi \ell$ and
$k\not\prec_\pi \ell$, while removing $\overline x_{j,\ell}$; examples
of this are $\ell$ equal to $\ell_1$ and, again,~$\ell_2$
in the same figure.

For a vertex~$\ell_3$ such that
$j\prec_\pi \ell_3$ and
$k\prec_\pi \ell_3$ but
$i\not\prec_\pi \ell_3$, this is done similarly to the inferences
(\ref{S2FormEq}) and~(\ref{S1FormEq})
but without the side literal~$\overline x_{j,k}$:
\begin{equation}\label{S4FormIonlyEq}
\AxiomC{\raisebox{7pt}{$S^\prime_4$}$\proofdots$}
\kernHyps{-1ex}
\noLine
\UnaryInfC{$\overline x_{i,j},\overline x_{j,\ell_3},\overline x_{\ell_3,i}$}
\AxiomC{$\proofdots$ \raisebox{1ex}{\hbox to 0pt{rest of~$S_4$}}}
\noLine
\UnaryInfC{$\overline x_{i,\ell_3},\overline x_{k,j}, \overline x_{i,j}, \overline \pi^*$}
\BinaryInfC{$\overline x_{j,\ell_3},\overline x_{k,j},\overline x_{i,j},\overline \pi^*$}
\DisplayProof
\end{equation}
Here $\overline \pi^*$ is
$\overline \pi_{-[jR(i\cap k);j\ell_3]}$.
The transitivity axiom~$T_{i,j,\ell_3}$ shown as the last line
of~$S_4^\prime$ is handled exactly as before.  This construction is
repeated for all such~$\ell_3$'s.

The vertices~$\ell_1$ such that
$j\prec_\pi \ell_1$ and
$i\prec_\pi \ell_1$ but
$k\not\prec_\pi \ell_1$ are handled in exactly the same way.
(The side literals~$\pi^*$ change each time to reflect the
literals that have already been replaced.)

Finally, consider a vertex~$\ell_2$ such that
$i\not\prec_\pi \ell_2$ and
$j\prec_\pi \ell_2$ and
$k\not\prec_\pi \ell_2$.  This is handled by the derivation
\begin{prooftree}
\AxiomC{\raisebox{7pt}{$S^{\prime\prime}_4$}$\proofdots$}
\kernHyps{-1ex}
\noLine
\UnaryInfC{$\overline x_{i,j},\overline x_{j,\ell_2},\overline x_{\ell_2,i}$}
\AxiomC{\raisebox{7pt}{$S^{\prime\prime\prime}_4$}$\proofdots$}
\kernHyps{-1ex}
\noLine
\UnaryInfC{$\overline x_{k,j},\overline x_{j,\ell_2},\overline x_{\ell_2,k}$}
\AxiomC{$\proofdots$ \raisebox{1ex}{\hbox to 0pt{rest of~$S_4$}}}
\noLine
\UnaryInfC{$\overline x_{i,j},\overline x_{i,\ell_2},\overline x_{k,j},\overline x_{k,\ell_2},\overline \pi^*$}
\BinaryInfC{$\overline x_{i,j},\overline x_{i,\ell_2},\overline x_{k,j},\overline x_{j,\ell_2},\overline \pi^*$}
\BinaryInfC{$\overline x_{i,j},\overline x_{k,j},\overline x_{j,\ell_2},\overline \pi^*$}
\end{prooftree}
As before, the set~$\pi^*$ of side literals is changed to reflect
the literals that have already been added and removed as~$S_4$ is being
created.
The subproofs $S^{\prime\prime}_4$ and~$S^{\prime\prime\prime}_4$
of the transitivity axioms $T_{i,j,\ell_2}$ and~$T_{k,j,\ell_2}$
are handled exactly as before, depending on the status of their
guard variables.

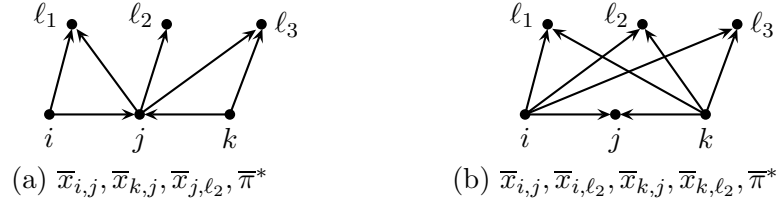
\begin{figure}[t]
\psset{unit=0.06cm}     
\hfill
\begin{pspicture}(16,-15)(72,24)
\pscircle*(20,0){2pt}
\pscircle*(40,0){2pt}
\pscircle*(60,0){2pt}
\pscircle*(25,20){2pt}
\pscircle*(46,20){2pt}
\pscircle*(67,20){2pt}
\psline[arrowscale=1.4 1.2]{->}(20,0)(25,19)
\psline[arrowscale=1.4 1.2]{->}(40,0)(26,19)
\psline[arrowscale=1.4 1.2]{->}(40,0)(46,19)
\psline[arrowscale=1.4 1.2]{->}(40,0)(65.8,18.8)
\psline[arrowscale=1.4 1.2]{->}(60,0)(67.2,19)
\psline[arrowscale=1.4 1.2]{->}(20,0)(39,0)
\psline[arrowscale=1.4 1.2]{->}(60,0)(41,0)
\uput[-90](20,0){$i$}
\uput[-90](40,0){$j$}
\uput[-90](60,0){$k$}
\uput[160](25,20){$\ell_1$}
\uput[160](46,20){$\ell_2$}
\uput[0](67,20){$\ell_3$}
\rput(40,-15){(a) $\overline x_{i,j},\overline x_{k,j},\overline x_{j,\ell_2},\overline \pi^*$}
\end{pspicture}
\hfill
\begin{pspicture}(16,-15)(72,24)
\pscircle*(20,0){2pt}
\pscircle*(40,0){2pt}
\pscircle*(60,0){2pt}
\pscircle*(25,20){2pt}
\pscircle*(46,20){2pt}
\pscircle*(67,20){2pt}
\psline[arrowscale=1.4 1.2]{->}(20,0)(25,19)
\psline[arrowscale=1.4 1.2]{->}(20,0)(45.5,19)
\psline[arrowscale=1.4 1.2]{->}(20,0)(66.0,19)
\psline[arrowscale=1.4 1.2]{->}(60,0)(26,19)
\psline[arrowscale=1.4 1.2]{->}(60,0)(46.5,19)
\psline[arrowscale=1.4 1.2]{->}(60,0)(67,19)
\psline[arrowscale=1.4 1.2]{->}(20,0)(39,0)
\psline[arrowscale=1.4 1.2]{->}(60,0)(41,0)
\uput[-90](20,0){$i$}
\uput[-90](40,0){$j$}
\uput[-90](60,0){$k$}
\uput[160](25,20){$\ell_1$}
\uput[160](46,20){$\ell_2$}
\uput[0](67,20){$\ell_3$}
\rput(40,-15){(b) $\overline x_{i,j},\overline x_{i,\ell_2},\overline x_{k,j},\overline x_{k,\ell_2},\overline \pi^*$}
\end{pspicture}
\hfill
\caption{The partial orders as changed by~$S_4$.}
\label{S4Fig}
\end{figure}

Finally, we describe how to form the subproof~$S_3$.  For this, we
must form the bipartite partial order~$\pi_3$ which is associated with
the partial order~$\tau(C_3)$,
where $C_3$ is the final clause of~$S_3$.
To obtain $\overline\pi_3$, we need to add
the literals $\overline x_{i,\ell}$ such that
$i\not\prec_\pi \ell$ and such that
either $ j\prec_\pi \ell$ or $k\prec_\pi \ell$, while removing
any literals $\overline x_{j,\ell}$ and $\overline x_{k,\ell}$.
This is done by exactly the same
construction used above in~(\ref{S4FormIonlyEq}).  The literals
in $\overline \pi_{-[jR(i);kR(i\cup j)]}$ are exactly the literals needed
to carry this out.  The construction is quite similar to the above
constructions, and we omit any further description.

That completes the description of how to
construct the LR partial refutations~$R_{i}$.
The process stops once some~$R_i$ has no unfinished clauses.
We claim that the process stops after polynomially many stages.

To prove this, recall that $R_{i+1}$ is formed
by handling the leftmost unfinished clause using
one of cases \caseiti{}-\caseitiv{}.  In the first three cases,
the unfinished clause is replaced by a derivation based on $P_\pi$
for some bipartite order~$\pi$.  Since $P_\pi$ has size $O(n^3)$, this
means that the number of clauses in~$R_{i+1}$ is at most the
number of clauses in~$R_i$ plus $O(n^3)$.  Also, by construction, $R_{i+1}$~has
one fewer unfinished clauses than~$R_i$.  In case~\caseitiv{} however,
$R_{i+1}$ is formed by adding up to $O(n)$ many clauses to~$R_i$
plus adding either two or three new unfinished leaf clauses.  In addition,
case~\caseitiv{} always causes at least one transitivity axiom~$T_{i,j,k}$
to be learned.
Therefore, case~\caseitiv{} can occur at most $2{n \choose 3} = O(n^3)$ times.
Consequently at most $3\cdot 2 {n \choose 3} = O(n^3)$ many
unfinished clauses are added throughout the entire process.
It follows that the process stops with~$R_i$ having no unfinished clauses
for some $i\le 6 {n \choose 3}=O(n^3)$.  Therefore there is a pool
refutation of $\GGT_n$ with $O(n^6)$ lines.  Since the $\GGT_n$ principle
has $O(n^3)$ many clauses, the number of inferences in
the refutation is bounded by a quadratic
polynomial of the number of the clauses being refuted.

By inspection, each clause in the refutation contains $O(n^2)$
literals.  This is because the largest clauses are those corresponding
to (small modifications of) bipartite partial orders, and because
bipartite partial orders can contain at most $O(n^2)$ many ordered pairs.
Furthermore, the refutations~$P_n$ for the graph tautology $\GT_n$ contain
only clauses of size $O(n^2)$.
\hfill \\
Q.E.D. Theorem~\ref{PoolResGgtThm} \hfill $\qed$
\end{proof}

Theorem~\ref{regRtiGgtThm} is proved with nearly the same construction.  In fact, the only change needed is the construction of~$S$ from~$P_\pi^\prime$.
Recall that in the proof of Theorem~\ref{PoolResGgtThm}, the
pool derivation~$S$ was formed
by using a depth-first traversal of~$P_\pi^\prime$.  This is not
sufficient for
Theorem~\ref{regRtiGgtThm}, since now the derivation~$S$ must use only input
lemmas.  Instead, we use Theorem~3.3 of~\cite{BHJ:ResTreeLearning}, which
states that a (regular) dag-like resolution derivation can
be transformed into a (regular) tree-like derivation with input lemmas.
Forming $S$ in this way from~$P^\prime_\pi$ suffices for the
proof of Theorem~\ref{regRtiGgtThm}: the lemmas of~$S$ are either
transitive closure axioms derived earlier in~$R_i$ or are
derived by input subproofs earlier in the post-ordering of~$S$.
Since the transitive closure axioms that appeared
earlier in~$R_i$ were derived by resolving two $\GGT_n$ axioms,
the lemmas used in~$S$ are all input lemmas.

The transformation of Theorem~3.3 of~\cite{BHJ:ResTreeLearning}
may multiply the size of the derivation
by the depth of the original derivation.  Since
it is possible to form the proofs~$P_\pi$
with depth $O(n)$, the overall size of the
regRTI refutation is $O(n^7)$.
This completes the proof of Theorem~\ref{regRtiGgtThm}.
\hfill $\qed$

\section{Greedy, unit-propagating DPLL with clause learning}\label{GreedySect}

This section discusses how the refutations in
Theorems \ref{PoolResGgtThm} and~\ref{regRtiGgtThm}
can be modified so as to ensure
that the refutations are greedy and unit-propagating.

\begin{definition}
Let $R$ be a
tree-like regular w-resolution refutation
with input lemmas.
Let $\Gamma(C)$ be
the set of clauses of $\Gamma$
plus every clause $D <_R C$ in~$R$ that has
been derived by an input subproof
and thus is available as a learned
clause to aid in the derivation of~$C$.

The refutation $R$ is {\em greedy and unit-propagating}
provided
that, for each clause~$C$ of~$R$, if
there is an input derivation
from~$\Gamma(C)$ of some clause $C^\prime\subseteq C^+$
which does not resolve on any literal in~$C^+$,
then $C$~is derived in~$R$ by such a
derivation.
\end{definition}

Note that, as proved in~\cite{BKS:clauselearning},
the condition that there is a input derivation from $\Gamma(C)$ of
some $C^\prime \subseteq C^+$ which does not resolve
on $C^+$ literals
is equivalent to the condition that if all literals of $C^+$
are set false
then unit propagation yields a contradiction from $\Gamma(C)$.
(In \cite{BKS:clauselearning}, these are called ``trivial'' proofs.)
This justifies the terminology ``unit-propagating''.

The definition of ``greedy and unit-propagating''
is actually a bit more restrictive than
necessary, since DPLL algorithms may actually learn multiple
clauses at once, and this can mean that $C$ is not derived from
a single input proof but rather from a combination of several input proofs
as described in the proof of Theorem 5.1 in~\cite{BHJ:ResTreeLearning}.

\begin{theorem}\label{GreedyRegRtiThm}
The guarded graph tautology formulas $\GGT_n$ have greedy, unit-propagating,
polynomial size, tree-like, regular w-resolution refutations with input lemmas.
\end{theorem}

\begin{proof}
We indicate how to modify the proofs
of Theorems \ref{PoolResGgtThm} and~\ref{regRtiGgtThm}.
We again build tree-like LR partial
refutations satisfying the same
properties a.-e.\ as before, except now w-resolution inferences
are permitted.  Instead of
being formed in distinct stages $R_0, R_1, R_2,\ldots$,
the w-resolution refutation~$R$ is constructed by one continuing process.
This construction incorporates all
of transformations \caseiti{}-\caseitiv{} and also incorporates
the construction of Theorem 3.3 of~\cite{BHJ:ResTreeLearning}.

At each point in the construction, we will be scanning
the so-far constructed partial w-resolution refutation~$R$
in preorder, namely in
depth-first, left-to-right order.  That is to say,
the construction recursively processes a node in the proof tree,
then its left subtree,
and then its right subtree.  During most steps of
the preorder scan,
the partial refutation~$R$ is modified by changing
the part that comes subsequently in the preorder,
but the construction may also add and remove literals
from clauses below the current clause~$C$.
When the preorder scan reaches a clause~$C$ that
has an input derivation~$R^\prime$ from $\Gamma(C)$
of some $C^\prime\subseteq C$
that does not
resolve on~$C^+$, then some such~$R^\prime$ is inserted
into~$R$ at that point.
When the preorder scan reaches an unfinished leaf~$C=C_0$,
then a
(possibly exponentially large) derivation~$P_\pi^*$ is added
as its derivation.  The construction continues processing~$R$
by scanning~$P_\pi^*$ in preorder, with the end result
that either (1)~$P_\pi^*$ is succcessfully processed and reduced
to only polynomial size or (2)~the preorder scan of~$P_\pi^*$
reaches a transitivity clause~$T_{i,j,k}$ of the type that
triggered case~\caseitiv{} of Theorem~\ref{PoolResGgtThm}.
In the latter case,
the preorder scan backs up to the root clause~$C_0$ of~$P_\pi^*$,
replaces $P_\pi^*$ with the
derivation~$S$ constructed in case~\caseitiv{}
of Theorem~\ref{PoolResGgtThm},
and restarts the preorder scan at clause~$C_0$.

We describe the actions of the preorder scan in more
detail. Initially, $R$ is the ``empty''
derivation, with the empty clause as its only
(unfinished) clause.  A clause~$C$ encountered during
the preorder scan of~$R$ is handled by one of
the following.

\begin{description}
\item[\caseitip] Suppose that some $C^\prime\subseteq C^+$
can be derived by an input
derivation from~$\Gamma(C)$ that does not resolve on any literals
of~$C^+$.
Fix any such~$C^\prime\subseteq C^+$,
and replace the subderivation
in~$R$ of the clause~$C$ with such a derivation
of~$C^\prime$ from $\Gamma(C)$.
Any extra literals in $C^\prime\setminus C$
are in $C^+$ and are propagated down
until reaching a clause where
they already appear, or occur as a phantom
literal.  There may also be
literals in $C \setminus C^\prime$: these
literals are removed as necessary from clauses
below~$C^\prime$ in~$R$ to maintain the property
of~$R$ containing correct w-resolution inferences.
Note that this can convert resolution inferences into
w-resolution inferences.
The clause $C^\prime$ is now a learned clause.

Note that this case includes
transitivity clauses $C=C^\prime=T_{i,j,k}$
that satisfy the conditions of cases \caseiti{}-\caseitiii{}
of Theorem~\ref{PoolResGgtThm}
\item[\caseitiip] If case~\caseitip{} does not apply,
and $C$ is not a leaf node, then $R$ is unchanged
at this point and the depth-first traversal proceeds
to the next clause.
\item[\caseitiiip] If $C$ is an unfinished clause
of the form $\falseBPOpi$, let $P_\pi$ be
as before.
Recall that no literal in~$C^+$ is resolved on in~$P_\pi$.
Unwind the proof~$P_\pi$
into a tree-like regular refutation~$P_\pi^*$
that is possibly exponentially big, and attach
$P_\pi^*$ to~$R$ as a proof of~$C$.  Mark the
position of~$C$ by setting $C_0=C$
in case it is necessary to
later backtrack to~$C$.  Then continue the preorder
scan by traversing into~$P_\pi^*$.

\item[\caseitivp] Otherwise, $C$ is a leaf clause of
the form $T_{i,j,k}$ and since case~\caseitip{} does not apply,
one of $T_{i,j,k}$'s guard literals~$x$,
namely $x=x_{r,s}$ or $x=\overline x_{r,s}$,
is in~$C^+$.  If $C$ is {\em not} inside
the most recently added~$P_\pi^*$ or if $x\in C_0^+$,
then replace
$T_{i,j,k}$ with $T_{i,j,k}\lor x$,
and propagate the literal~$x$
downward in the refutation until reaching a
clause where it appears as a literal or a phantom literal.
Otherwise, the preorder scan backtracks to the
root clause~$C_0$ of $P_\pi^*$, and replaces $P_\pi^*$ with
the partial resolution refutation~$S$ formed in
case~\caseitiv{} of Theorem~\ref{PoolResGgtThm}.
\end{description}

It is clear that this process eventually halts with
a valid greedy, unit-propagating, tree-like w-resolution
refutation.  We claim that it also yields
a polynomial size refutation.  Consider
what happens when a derivation~$P_\pi^*$ is inserted.
If case~\caseitivp{} is triggered, then the proof~$S$ is
inserted in place of~$P_\pi^*$, so the size
of~$P_\pi^*$ does not matter.
If case~\caseitivp{} is not triggered,
then, as in the proof of Theorem 3.3 of~\cite{BHJ:ResTreeLearning},
the preorder scan of~$P_\pi^*$ modifies (the possibly
exponentially large) $P_\pi^*$ to
have polynomial size.
Indeed, as argued in~\cite{BHJ:ResTreeLearning},
any clause~$C$
in~$P_\pi^*$ will occur at most $d_C$~times in the modified
version of~$P_\pi^*$ where $d_C$ is the depth of the derivation
of~$C$ in the original~$P_\pi$.  This is because $C$ will
have been learned by an input derivation once it has appeared no more than~$d_C$
times in the modified derivation~$P_\pi^*$.  This is
proved by induction on~$d_C$.

Consider the situation where $S$ has
just been inserted in place
of~$P_\pi^*$ in case~\caseitiiip{}.
The transitivity clause~$T_{i,j,k}$
is not yet learned at this point,
since otherwise case~\caseitip{} would have
applied.  We claim, however, that $T_{i,j,k}$
is learned as~$S$ is traversed.  To prove this,
since $T_{i,j,k}$ is manifestly derived
by an input derivation and since its guard
literals $x_{r,s}$ and~$\overline x_{r,s}$ do
not appear in~$C_0^+$, it is
enough to show that the clause~$T_{i,j,k}$ is
reached in the preorder traversal scan of~$S$.
This, however, is an immediate consequence of the fact that
$T_{i,j,k}$ was reached in the preorder scan of~$P_\pi^*$
and triggered case~\caseitivp{},
since if case~\caseitip{} applies to $T_{i,j,k}$
or to any clause below $T_{i,j,k}$
in the preorder scan of~$S$,
then it certainly also applies $T_{i,j,k}$ or
some clause below $T_{i,j,k}$
in the preorder scan of~$P_\pi^*$.

The size of the final refutation~$R$ is bounded
the same way as in the proof of Theorem~\ref{regRtiGgtThm},
and this completes the proof of Theorem~\ref{GreedyRegRtiThm}.
\hfill $\qed$
\end{proof}

\begin{theorem}\label{DPLLpolySizeThm}
There are DPLL search procedures with clause
learning which are greedy, unit-propagating, but do not use
restarts, that
refute the $\GGT_n$ clauses in polynomial time.
\end{theorem}

We give a sketch of the
proof.
The construction for the proof of Theorem~\ref{GreedyRegRtiThm}
requires only that the clauses~$T_{i,j,k}$ are learned
whenever possible, and does not depend on whether any other
clauses are learned.
This means that the following
algorithm for DPLL search with clause learning
will always succeed in finding a refutation of
the $\GGT_n$ clauses: At each point, there is
a partial assignment~$\tau$.
The search algorithm must do one of the following:
\setlength{\parsep}{0in}
\setlength{\itemsep}{0in}
\begin{description}
\item[\rm (1)] If unit propagation yields a
contradiction, then learn a clause $T_{i,j,k}$ if
possible, and backtrack.
\item[\rm (2)] Otherwise, if there are
any literals in the transitive closure of
the bipartite partial order
associated with~$\tau$ which are not assigned
a value,
branch on one of these literals to set
its value.  (One of the true or false assignments
yields an immediate
conflict, and may allow learning a clause~$T_{i,j,k}$.)
\item[\rm (3)] Otherwise,
determine whether there is a clause~$T_{i,j,k}$
which is used in the proof~$P_\pi$ whose guard literals
are resolved on in~$P_\pi$.  (See Lemma~\ref{BpoDerivationLm}.)
If not, do a DPLL traversal of~$P_\pi$,
eventually backtracking from the assignment~$\tau$.
\item[\rm (4)] Otherwise, some clause
$T_{i,j,k}$ blocks
$P_\pi$ from being traversed in polynomial time.  Branch on
its variables in the order given
in the proof of Theorem~\ref{PoolResGgtThm}.
From this, learn the clause
$T_{i,j,k}$.
\end{description}

\section{Guarded, xor-ified, pebbling tautologies} \label{sec:GPeb}
This section gives polynomial size regRTI refutations
for the
guarded pebbling tautologies which
Urquhart~\cite{Urquhart:regularresolution} proved require
exponential size regular resolution proofs.

\begin{definition}
A {\em pointed dag}~$G=(V,E)$ is a directed acyclic graph with a
single sink~$t$ such that every vertex in~$G$ has
indegree either $0$ or~$2$.
The pebbling formula $\Peb(G)$
for a pointed dag~$G$
is the unsatisfiable formula in the variables~$x_v$ for $v\in V$
consisting of the following clauses:
\begin{itemize}
\setlength{\parsep}{0pt}
\setlength{\itemsep}{0pt}
\item[$(\alpha)$] $x_s$, for every source $s \in V$,
\item[$(\beta)$] $\negx_u\lor\negx_v \lor x_w$,
for every vertex~$w$ with two (immediate) predecessors $u$ and~$v$,
\item[$(\gamma)$] $\negx_t$, for $t$ the sink vertex.
\end{itemize}
\end{definition}
We next define Urquhart's ``xor-ification'' of a pebbling tautology
clause.  Xor-ification, for two variables,
is due originally to Alekhnovich
and Razborov, see Ben-Sasson~\cite{BenSasson:sizespace}, and is similar
to the ``orification'' used by~\cite{BIW:nearoptimal}.
The intuition for xor-ification is that each variable~$x_u$ is
replaced by a set of clauses which
express the $k$-fold exclusive or
$x_{u,1}\oplus \cdots\oplus x_{u,k}$.
\begin{definition}
Let $k>0$, and $x_u$ be a variable of the $\Peb(G)$ principle.
Let $x_u^1$ be $x_u$, and $x_u^{-1}$ be its complement~$\negx_u$.
Define $x_u^\koplus$ to be the set of clauses of the form
\begin{equation}\label{eq:xorDefn}
x_{u,1}^{i_1} \lor x_{u,2}^{i_2} \lor \cdots \lor x_{u,k}^{i_k}
\end{equation}
where an even number of the values $i_j$ equal~$-1$ (and the rest equal~$1$).
Dually, define $\negx_u ^\koplus$ to be the set of
clauses of the form (\ref{eq:xorDefn}) with an odd number of the $i_j$
equal to~$-1$.
Note there are $2^{k-1}$ clauses in each of $x_u^\koplus$
and~$\negx_u^\koplus$.
If $C$ is a clause
$C = z_1\lor \cdots \lor z_\ell$, each $z_i$ of the form $x_u$ or~$\negx_u$,
then $C^\koplus$ is the
set of clauses of the form
\[
C_1 \lor C_2 \lor \cdots \lor C_\ell,
\]
where each $C_i \in z_i^\koplus$.  There are $2^{(k-1)\ell}$
many clauses in~$C^\koplus$.
\end{definition}
\begin{definition}
The {\em xor-ified pebbling formula $\Peb^\koplus(G)$}
is the set of clauses $C^\koplus$ for $C\in\Peb(G)$.
\end{definition}
\begin{definition}
Let $G$ be a pointed graph with $n$~vertices and $k=k(n)>0$.
Let $\rho$ be a function with domain the set of clauses
of $\Peb^\koplus(G)$ and with range the set of variables~$x_{u,i}$
of $\Peb^\koplus(G)$, such that, for all~$C$,
the variable $\rho(C)$ is not used in~$C$.
The {\em guarded xor-ified pebbling formula $\GPeb^\koplus(G)$}
is the set of clauses of the forms
\[
C\lor \rho(C) \qquad\qquad \hbox{and} \qquad\qquad
C \lor \overline{\rho(C)}
\]
for $C\in\Peb^\koplus(G)$.
\end{definition}
Note that, as in the case of $\GGT_n$,
the $\GPeb^\koplus(G)$ clauses depend on
the choice of~$\rho$; again this is suppressed in the notation.
For a dag $G$ with $n$ vertices, the formula $\GPeb^\koplus(G)$
consists of $O(2^{3k}n)$ clauses.

Our definitions of $\Peb^\koplus(G)$ and $\GPeb^\koplus(G)$ differ
somewhat from Urquhart's, but these differences are inessential and
make no difference to asymptotic proof sizes.

Of course, the $\Peb^\koplus(G)$ clauses
are readily derivable from the $\GPeb^\koplus(G)$
clauses by resolving on the guard literals as given
by~$\rho$.  There are simple polynomial size regular resolution
refutations of the $\Peb^\koplus(G)$ principles; hence
there are polynomial size, but not regular, resolution
refutations of the $\GPeb^\koplus(G)$ principles.
Indeed,
Urquhart~\cite{Urquhart:regularresolution}
proved that there are pointed graphs~$G$ with $n$ vertices
and values $k=k(n)=O(\log\log n)$, and
functions~$\rho$, such
that regular resolution refutations
of the $\GPeb^\koplus(G)$ clauses require
size $2^{\Omega(n/((\log n)^2 \log\log n))}$.

We shall show that the $\GPeb^\koplus(G)$ principles have polynomial size
proofs in regRTI and in pool resolution:
\begin{theorem}\label{GPebProofsThm}
The guarded xor-ified pebbling formulas $\GPeb^{\koplus}(G)$
have polynomial size regRTI refutations, and thus polynomial
size pool refutations.
\end{theorem}
We conjecture that, analogously to Theorem~\ref{DPLLpolySizeThm},
the $\GPeb^\koplus(G)$ principles can be shown
unsatisfiable by polynomial size, greedy, unit propagating
DPLL clause learning; however, we have not attempted to prove this.

We make some simple observations
about working with xor-ified clauses
before proving Theorem~\ref{GPebProofsThm}.
\begin{lemma}\label{XorProofContraLm}
Let $u$ be a vertex in~$G$.
There is a tree-like regular
refutation of the clauses in $x_u^\koplus$ and~$\negx_u^\koplus$
with $2^k-1$ resolution inferences.  Its resolution variables
are the variables~$x_{u,i}$.
\end{lemma}
\begin{proof}
This is immediate by inspection: the refutation consists of resolving
on the literals $x_{u,i}$ successively for $i=1,2,\ldots,k$,
giving a proof of height~$k$.
The leaf clauses of the proof are the members of
$x_u^\koplus$ and $\negx_u^\koplus$.
\hfill $\qed$
\end{proof}
The refutation of Lemma~\ref{XorProofContraLm} can be viewed
as being the ``$\koplus$-translation'' of the proof
\begin{prooftree}
\AxiomC{$x_u$}
\AxiomC{$\negx_u$}
\BinaryInfC{$\bot$}
\end{prooftree}
The next lemma describes a similar ``$\koplus$-translation''
of a proof
\begin{prooftree}
\AxiomC{$C,x_u$}
\AxiomC{$D,\negx_u$}
\BinaryInfC{$C,D$}
\end{prooftree}
\begin{lemma}\label{XorProofContraSideLm}
Let $u$ be a vertex in~$G$, and let $C$ and~$D$ be clauses
which do not contain either
$x_u$ and~$\negx_u$.
Then each clause of $(C\lor D)^\koplus$ has a tree-like regular derivation
from the clauses in $(C\lor x_u)^\koplus$ and $(D\lor \negx_u)^\koplus$
in which the variables used as resolution variables
are exactly the variables $x_{u,i}$. This derivation has $2^k-1$
resolution inferences, and $2^k$ leaf clauses.
\end{lemma}
\begin{proof}
Fix a clause~$E$ from $(C\lor D)^\koplus$; we must
describe its derivation from clauses in
$(C\lor x_u)^\koplus$ and $(D\lor \negx_u)^\koplus$.
Let $E_C$ be the subclause of~$E$ which is
from $C^\koplus$, and let $E_D$ the subclause
of~$E$ which is from $D^\koplus$.  If $C$ and~$D$
have non-empty intersection, $E_C$ and~$E_D$ are not disjoint;
however, in any event, $E = E_C\cup E_D$.

Form the refutation from Lemma~\ref{XorProofContraLm}.  Then
add $E_C$ to every leaf clause from~$x_u^\oplus$, add $E_D$
to every leaf clause from~$\negx_u^\koplus$, and
add $E$ to every non-leaf clause.  This gives the desired
derivation of~$E$.
\hfill $\qed$.
\end{proof}
Lemma~\ref{XorProofContraSideLm}
lets us further generalize the construction of
$\koplus$-translations of proofs. As a typical example,
the next lemma gives the $\koplus$-translation
of the derivation
\begin{prooftree}
\AxiomC{$\negx_u,\negx_v, x_w$}
\AxiomC{$x_u$}
\BinaryInfC{$\negx_v,x_w$}
\AxiomC{$x_v$}
\BinaryInfC{$x_w$}
\end{prooftree}
\begin{lemma}\label{XorUVimpliesWLm}
Let $w$ be a vertex of~$G$, and $u$ and~$v$ its
predecessors.
Then, each clause in~$ x_w^\koplus$ has a dag-like resolution
derivation~$P$ from the clauses
in $x_u^\koplus$, $x_v^\koplus$,
and $(\negx_u \lor \negx_v\lor x_w)^\koplus$.
This derivation contains $<2^{2k}$ resolution inferences
and resolves on the literals $x_{u,i}$ and $x_{v,i}$.
In addition, the paths in~$P$ that lead to clauses
in $x_v^\koplus$ resolve on exactly the literals $x_{v,i}$.
\end{lemma}
Lemma~\ref{XorUVimpliesWLm} follows by applying
Lemma~\ref{XorProofContraSideLm} twice. \hfill $\qed$
\smallskip

It is important to note that the left-to-right order
of the leaves of the
derivation of Lemma~\ref{XorUVimpliesWLm} can
be altered by changing the left-right
order of hypotheses
of resolution inferences.  In particular, given any
leaf clause~$D$ of a refutation~$P$, we can
order the hypotheses of the resolution inferences
so that $D$ is the leftmost leaf clause.  This will
be useful when $D$ needs to be
learned.

\begin{definition}
$G\rest w$ is the induced subgraph of~$G$ with sink~$w$ and
containing those vertices from which the vertex~$w$ is reachable.
$G[w]$ is the graph resulting from~$G$ by making the vertex~$w$
a leaf by removing its incoming edges, and then removing those vertices
from which the sink vertex of~$G$ is no longer reachable.

The vertex~$u$ is an {\em ancestor} of~$w$ if $u\not= w$
and $u\in G\rest w$.
We call $u$ and $v$ {\em independent ancestors} of~$w$
provided $u$, $v$, and $w$ are distinct
and $u\in (G\rest w)[v]$ and $v\in(G\rest w)[u]$.
In this case, we write
$G[u,v]$ for $G[u][v]=G[v][u]$.  Sometimes $v$, and possibly also~$u$,
may be missing or undefined; in these cases, $G[u,v]$
means just $G$ if both $u$ and~$v$ are undefined, and
means $G[u]$ if $u$ is present but $v$ is undefined.
\end{definition}

Note that it is possible for $u$ and~$v$ to be
independent ancestors of~$w$, and also have $u$
an ancestor of~$v$ or vice-versa.

The polynomial size regular resolution refutations of the
$\Peb^\koplus(G)$ principles also apply to
subgraphs such as $G\rest w$, $(G\rest w)[u]$
and $(G\rest w)[u,v]$.  We can use these
refutations to prove
the following lemma.
We write $\Peb_\alphabeta^\koplus(G)$ to
denote the $\koplus$-translations
of $\Peb(G)$ clauses of type ($\alpha$) and~($\beta$), omitting
the clauses of type~($\gamma$),
and similarly for $\GPeb_\alphabeta^\koplus(G)$.

\begin{lemma}\label{RegXorPebRefuteLm}
\begin{description}
\item[\rm (i)] Let $w$ be a vertex of~$G$.  Then
each clause in~$x_w^\koplus$ has
a regular resolution derivation from
the clauses $\Peb_\alphabeta^\koplus(G\rest w)$.
The derivation uses only
resolution variables of the
form~$x_{a,i}$ with $a\in (G\rest w) \setminus\{w\}$.
\item[\rm (ii)] Let $u$ and~$w$ be distinct vertices of~$G$
such that $u$ is an ancestor of~$w$.  Then
each clause in $(\negx_u \lor x_w)^\koplus$
has a regular resolution derivation from
the clauses $\Peb_\alphabeta^\koplus((G\rest w)[u])$.
The derivation uses only resolution variables of the
form~$x_{a,i}$ for
$a\in ((G\rest w)[u])\setminus\{u,w\}$.
\item[\rm (iii)] Let $u$ and~$v$ be
independent ancestors of~$w$.
Then each clause of $(\negx_u\lor \negx_v\lor x_w)^\koplus$
has a regular resolution derivation from
the clauses $\Peb_\alphabeta^\koplus((G\rest w)[u,v])$.
The derivation uses only resolution variables of
the form $x_{a,i}$ for
$a\in ((G\rest w)[u,v])\setminus\{u,v,w\}$.
\end{description}
\noindent
In all three cases, the derivation is dag-like and
has size $O(2^{3k}n)$ and height $O(kn)$.
\end{lemma}

\begin{proof}
We describe regular refutations from $\Peb(G)$.
Their $\koplus$-translations will give the regular
derivations of Lemma~\ref{RegXorPebRefuteLm}.

For~(i), there is an obvious regular dag-like, size $O(n)$,
derivation~$P$ of~$x_w$ from the clauses ($\alpha$) and~($\beta$)
of the (non-xorified)
principle $\Peb(G)$; the derivation proceeds by resolving on literals~$x_u$
in a depth-first traversal of~$G$.  Forming the
$\koplus$-translation $P^\koplus$ of~$P$ forms the desired derivation for part~(i)
of any given clause of~$x_w^\koplus$.
Since $P$ has size $O(n)$ and clauses in~$P$ have at most three literals,
$P^\koplus$ has size $O(2^{3k}n)$ and height $O(kn)$.
The limitation on which variables can be used as resolution
variables follows by inspection; it is also a consequence
of the fact that the refutation is
regular.

The proofs of (ii) and~(iii) are similar.
\hfill $\qed$
\end{proof}

\begin{proof} (of Theorem~\ref{GPebProofsThm}.)
We again construct a finite
sequence of ``LR partial refutations'',
denoted $R_0, R_1, R_2, \ldots$.  This terminates
after finitely many steps with the desired refutation~$R$.
Each LR partial refutation~$R_i$ will be a correct
regRTI proof (and thus a correct pool resolution refutation);
furthermore, it will satisfy the following conditions:
\begin{description}
\item[\rm a.] $R_i$ is a tree of nodes labeled with clauses.
The root is labeled
with the empty clause.  Each non-leaf node in~$R_i$ has a left
child and a right child, and the clauses labeling these nodes
form a valid resolution inference.
\item[\rm b.] For each clause~$C$ in~$R_i$, the clause~$C^+$
is defined as before as
\begin{eqnarray*}
C^+ & := & \{\ell : \hbox{The literal~$\ell$
occurs in some clause on the branch} \\
& & \quad \quad \quad \quad \quad
\hbox{from the root node of~$R_i$ up to and including~$C$}\}.
\end{eqnarray*}
\item[\rm c.] Each leaf of~$R_i$ is either ``finished'' or ``unfinished''.
Each finished node leaf~$L$ is labeled with either a
clause from $\GPeb^\koplus(G)$ or with a clause that was derived by
an input subderivation of~$R_i$ to the left of~$L$ in postorder.
The input subderivation
may not contain any unfinished leaves.
\item[\rm d.] Each unfinished leaf is labeled with
a clause~$C \in E^\koplus$ for a clause~$E$
such that one of the following three possibilities I.-III.{} holds.
I.~$E$~is of the form~$x_w$, and $C+$~contains no
literal $x_{a,i}$ for vertex
$a\in (G \rest w)\setminus\{w\}$.
II.~$E$~is of the form $\negx_u\lor x_w$ with $u$ an ancestor of~$w$,
and $C+$~contains no
literal $x_{a,i}$ for vertex
$a\in (G \rest w)[u]\setminus\{u,w\}$. Or,
III.~$E$~is of the form $\negx_u\lor \negx_v\lor x_w$
with $u$ and~$v$ independent ancestors of~$w$,
and $C^+$ contains no literal
$x_{a,i}$ for any vertex
$a\in (G \rest w)[u,v]\setminus\{u,v,w\}$.
\end{description}

We introduce a new notational
convention to describe (sub)clauses in~$R_i$.  For $w$ a vertex
in~$G$, the notation $W$ or $W^\prime$ denotes a
clause in $x_w^\koplus$, and $\negW$ or~$\negW^\prime$ to denotes
a clause in $\negx_w^\koplus$.  The notation $\negW$ or $\negW^\prime$
in no way denotes the negation of $W$ or~$W^\prime$; instead, they
are names of clauses, with the overline
meant only to serve as a reminder of the semantic meaning.

The initial refutation~$R_0$
is formed as follows.
Let $Q$ be the $\koplus$-translation of the
inference
\begin{prooftree}
\AxiomC{$\negx_t$}
\AxiomC{$x_t$}
\BinaryInfC{$\bot$}
\end{prooftree}
as given by Lemma~\ref{XorProofContraLm}, where $t$ is
the sink of~$G$.  There are $2^k$
leaf clauses of~$Q$:
half of them are labeled with a clause $\negT\in \negx_t^\koplus$
and the other half are labeled with a clause $T\in x_t^\koplus$.
Form $R_0$ from~$Q$ by replacing each leaf clause~$\negT$ with
a derivation
\begin{prooftree}
\AxiomC{$\negT, \rho(\negT)$}
\AxiomC{$\negT,\overline{\rho(\negT)}$}
\BinaryInfC{$\negT$}
\end{prooftree}
These inferences are regular,
since $\rho(\negT)$ is not an~$x_{t,i}$.
The other leaf
clauses, of the form~$T$, satisfy condition~d. and are unfinished
clauses in~$R_0$.

For the inductive step, the LR partial refutation~$R_i$
will be transformed into~$R_{i+1}$.  There are several
cases to consider; the goal is to replace
one unfinished leaf by a derivation containing only
finished leaves, or to learn one more $\Peb^\koplus(G)$
clause while
adding only polynomially many more unfinished leaves.

Consider the leftmost unfinished leaf of~$R_i$. By condition~d.,
its clause~$C$
will have one of the forms $W$, or $\negU,W$, or $\negU,\negV,W$
where $\negU\in\negx_u^\koplus$, $\negV\in \negx_v^\koplus$, and
$W\in x_w^\koplus$.
By Lemma~\ref{RegXorPebRefuteLm},
there is a dag-like regular refutation~$P$ of~$C$ from the clauses
of $\Peb_\alphabeta^\koplus((G\rest w)[u,v])$.  We wish to convert~$P$
into a derivation from the clauses of
$\GPeb_\alphabeta^\koplus((G\rest w)[u,v])$ and the already learned
clauses of~$R_i$.   Consider
each leaf clause~$D$ of~$P$.  Then $D$ is a $\koplus$-translation
of a clause in~$\Peb_\alphabeta^\koplus((G\rest w)[u,v])$.
As in the proof of Theorem~\ref{PoolResGgtThm}, there are four cases
to consider:
\begin{description}
\item[\caseiti] If the clause~$D$ is already learned as an input
lemma in~$R_i$ to the left of~$C$, then $D$ may be used in~$P$ as is.
\end{description}
For the remaining cases, assume $D$ has not been learned as an input lemma.
\begin{description}
\item[\caseitii]  Let $y = \rho(D)$.  If either $y$ or~$\negy$
is a member of~$C^+$, then add that literal to~$D$ and every clause
on the path below~$D$ until reaching the first clause where it appears.
This replaces $D$ with the $\GPeb_\alphabeta^\koplus((G\rest w)[u,v])$
clause $D\lor y$ or $D\lor\negy$.
\item[\caseitiii] Suppose cases \caseiti{} and~\caseitii{} do not apply
and that $y$ is not used as a resolution variable below~$D$.
In this case, replace $D$ by a resolution inference deriving
$D$ from $D\lor y$ and $D\lor \negy$.  This preserves
the regularity of the derivation. It also makes $D$ a learned clause.
\end{description}
It is possible that $C$ itself is a $\Peb_\alphabeta^\koplus(G)$ clause.
If so, then $C=D$ and $P$ is the trivial derivation containing only~$C$, and
one of cases \caseiti{}-\caseitiii{} holds.

If all leaf clauses~$D$ of $P$ can be treated by cases
\caseiti{}-\caseitiii{}, then we have successfully transformed~$P$
into a (still dag-like) derivation~$P^\prime$ which satisfies
regularity and in which leaf clauses are from $\GPeb^\koplus(G)$ or
already learned as input lemmas in~$R_i$.  By Theorem~3.3
of~\cite{BHJ:ResTreeLearning}, $P^\prime$~can be converted in a regRTI
proof~$P^\pprime$ of the same conclusion as~$P$, preserving the
regularity conditions, and with the size of~$P^\prime$ bounded by
twice the product of the size of~$P$ and the height of~$P$.
Therefore, the size of~$P^\pprime$ is $O((2^{3k}n)(kn)) = O(k 2^{3k}
n^2)$.  Form~$R_{i+1}$ by replacing the clause~$C$ in~$R_i$ with the
derivation~$P^\pprime$.  $R_{i+1}$~satisfies conditions a.-d., and has
one fewer unfinished clauses than~$R_i$.

However, if even one leaf clause~$D$ of $P$ fails
cases \caseiti{}-\caseitiii{}, then
the entire subderivation~$P$ is abandoned,
and we chose some leaf clause~$D$
of~$P$ to be learned
such that $D$ does not fall into cases \caseiti{}-\caseitiii.

The leaf clause~$D$ is
the $\koplus$-translation of an ($\alpha$) or~($\beta$)
clause of $\Peb^\koplus(G)$ and thus
either has the form $E\in x_e^\oplus$ for
some source~$e$ in~$G$ or
has the form $\negA,\negB,E$ where
$\negA \in \negx_a^\koplus$, $\negB \in \negx_b^\koplus$, and
$E\in x_e^\koplus$ for $a$, $b$, and $e$ vertices in~$G\rest w$ with
$a$ and~$b$ the two predecessors of~$e$ in~$G$.  Without
loss of generality, $b$~is not an ancestor of~$a$ in~$G$; otherwise
interchange $a$ and~$b$.  The construction now splits into three cases
I., II., and~III.{} depending on the form of~$C$.
These three cases each split into two subcases depending whether
$D$ is $E$ or is $\negA,\negB,E$.

I.\ In the first case, $C$ is equal to just $W$.  Recall
that $W$ is a clause in~$x_w^\koplus$.  First suppose
$D$ is equal to~$E$.  Consider the
derivation structure
\begin{equation}\label{GPebProofEqAalpha}
\AxiomC{$x_e$}
\AxiomC{$\negx_e, W$}
\BinaryInfC{$W$}
\DisplayProof
\end{equation}
Note that (\ref{GPebProofEqAalpha}) contains a blend of
variables
from $\Peb(G)$ (non-xorified) and from $\Peb^\koplus(G)$ (xor-ified).
However, we can still form its $\koplus$-translation~$Q$:
the leaf clauses of~$Q$ are the $2^k$ clauses of the
form $E^\prime \in x_e^\koplus$  and of the form
$\overline E^\prime, W$ for $\overline E^\prime\in \negx_e^\koplus$.
By choosing the appropriate left-right order for the
hypotheses in~$Q$, we arrange for $D=E$ to be the leftmost leaf
clause of~$Q$. Let $y = \rho(D)$.  Then $y$ is not one of the variables~$x_{e,i}$,
nor is $y$ or $\negy$ in~$C^+$ since condition~\caseitii{} does not hold for~$D$.
Therefore, we can modify~$Q$ by replacing $D=E$ with
\begin{equation}\label{LearnDclauseEq}
\AxiomC{$D,\rho(D)$}
\AxiomC{$D,\overline{\rho(D)}$}
\BinaryInfC{$D$}
\DisplayProof
\end{equation}
Form $R_{i+1}$ from~$R_i$ by replacing $C$ with the modified~$Q$.
This causes $D$ to become learned as an input lemma in~$R_{i+1}$.
The other leaf clauses of~$Q$ all satisfy condition~d.{} in~$R_{i+1}$
and thus become unfinished clauses of~$R_{i+1}$: they are all to the right
of~$D$.  This adds fewer than $2^k$ new unfinished clauses to~$R_{i+1}$.

Second suppose $D$ equals $\negA,\negB,E$.
For the moment, assume $e\not= w$.
Consider the following
derivation structure:
\begin{equation}\label{GPebProofEqA}
\AxiomC{$\negx_a,\negx_b, x_e$}
\AxiomC{$x_a$}
\BinaryInfC{$\negx_b,x_e$}
\AxiomC{$x_b$}
\BinaryInfC{$x_e$}
\AxiomC{$\negx_e,W$}
\BinaryInfC{$W$}
\DisplayProof
\end{equation}
Define $Q$ to be the $\koplus$-translation
of~(\ref{GPebProofEqA}), so $Q$ is a
tree-like derivation of the clause~$W$ from the clauses
of the form $\negA^\prime,\negB^\prime,E^\prime$,
the form $A^\prime$, the form $B^\prime$,
and the form $\negE^\prime,W$, where the clauses
$A^\prime$, $\negA^\prime$, $B^\prime$, $\negB^\prime$,
$E^\prime$, and $\negE^\prime$ range over all
members of $x_a^\koplus$, $\negx_a^\koplus$,
$x_b^\koplus$, $\negx_b^\koplus$, $x_e^\koplus$,
and $\negx_e^\koplus$, respectively.

The left-right order of hypotheses in~$Q$ is chosen so
that the clause~$D$ is the leftmost leaf clause of~$Q$.
The variable $y=\rho(D)$
is not one of the $x_{a,i}$'s, $x_{b,i}$'s, or $x_{e,i}$'s.
Since condition~\caseitii{} does not hold,
neither $y$ nor~$\negy$ is in~$C^+$.
Therefore, we again replace the clause~$D$ in~$Q$ with the
inference~(\ref{LearnDclauseEq}).
Form $R_{i+1}$ by replacing
the unfinished clause~$C$ in~$R_i$ with the
derivation~$Q$.  The inference (\ref{LearnDclauseEq})
does not violate regularity, and causes $D$ to be
learned as an input lemma in~$R_{i+1}$.

It is easy to check that the remaining leaf
clauses of~$Q$, which have the forms $\negA^\prime$, and
$\negB^\prime$, and $\negE^\prime,W$,
and $\negA^\prime,\negB^\prime,E^\prime$,
satisfy condition~d.
These are thus valid as new unfinished clauses in~$R_{i+1}$.
A simple calculation shows that there are
$2\cdot 2^{3(k-1)}+2^{2(k-1)}+2^{k-1} < 2^{3k}$ many of these clauses.

On the other hand, suppose $e=w$.  In this case, use
the derivation structure
\begin{prooftree}
\AxiomC{$\negx_a,\negx_b, W$}
\AxiomC{$x_a$}
\BinaryInfC{$\negx_b,W$}
\AxiomC{$x_b$}
\BinaryInfC{$W$}
\end{prooftree}
instead of~(\ref{GPebProofEqA}).  Let $Q$ be the
$\koplus$-translation of this, and argue again as
above.  We leave the details to the reader.

This completes the construction of~$R_{i+1}$ in
the case where $C$ is just~$W$.  At least one new
clause,~$D$, from $\Peb^\koplus(G)$ has been learned as
an input lemma. Since $D$ is learned in the
leftmost branch of~$Q$, it is
available for use as a learned clause for all future
unfinished clauses.
Fewer than $2^{3k}$ many new unfinished
leaves have been introduced.

II.\ Now consider the case where $C$ has the form $\negU,W$.
First suppose $D$ is equal to $E\in x_e^\oplus$, where $e$
is a source of~$G$.  We have $e\not= u$ since, even if $u$
is a leaf, $P$~does
not use the axiom of type~$(\alpha)$ for~$U$.  Thus $e$ and~$u$
are independent ancestors of~$w$.
Consider the derivation structure
\begin{equation}\label{GPebProofEqBalpha}
\AxiomC{$x_e$}
\AxiomC{$\negx_e,\negU,W$}
\BinaryInfC{$\negU, W$}
\DisplayProof
\end{equation}
Let $Q$ be the $\koplus$-translation of~(\ref{GPebProofEqBalpha})
arranged so that $D$ is its leftmost leaf clause.
Now argue as in the earlier cases to form~$R_{i+1}$.

Second, suppose $D$ is equal to $\negA,\negB,E$.
We suppose that $e\not= w$, and leave the (slightly simpler)
case of $e=w$ to the reader.\footnote{The idea is
to modify~(\ref{GPebProofEqB}) by omitting its final
inference and changing the $x_e$'s to~$W$'s.}
For the moment, suppose $u$ is not equal to $a$ or~$b$.  Again, it
is not possible for $u$ to equal~$e$, since $P$ does
not use the clause of type~($\beta$) for $u$ and its predecessors.
Consider the derivation structure
\begin{equation}\label{GPebProofEqB}
\AxiomC{$\negx_a,\negx_b, x_e$}
\AxiomC{$\qmark_1,x_a$}
\BinaryInfC{$\qmark_1,\negx_b,x_e$}
\AxiomC{$\qmark_2,x_b$}
\BinaryInfC{$\qmark_1,\qmark_2,x_e$}
\AxiomC{$\qmark_3,\negx_e,W$}
\BinaryInfC{$\negU,W$}
\DisplayProof
\end{equation}
Here the subclauses $\qmark_1$, $\qmark_2$, $\qmark_3$ are unknown and must be
determined: in fact, one of them will be $\negU$ and the
other two will be empty.  The rules for
determining these are as follows. First, if $u$ is an ancestor
of~$a$ (and thus not equal to~$a$), select $\qmark_1$ to equal~$\negU$.
Second and otherwise, if $u$ is an ancestor of~$b$,
select $\qmark_2$ to equal~$\negU$.  Otherwise,
$u$ and~$e$ are independent ancestors of~$w$, so
select $\qmark_3$ to equal~$\negU$.  For example, in the second case,
the derivation~(\ref{GPebProofEqB}) becomes
\[
\AxiomC{$\negx_a,\negx_b, x_e$}
\AxiomC{$x_a$}
\BinaryInfC{$\negx_b,x_e$}
\AxiomC{$\negU,x_b$}
\BinaryInfC{$\negU,x_e$}
\AxiomC{$\negx_e,W$}
\BinaryInfC{$\negU,W$}
\DisplayProof
\]
Now let $Q$ be the $\koplus$-translation of~(\ref{GPebProofEqB}),
with $D$ the leftmost leaf clause of~$Q$.
Modify $Q$ by replacing~$D$ with~(\ref{LearnDclauseEq}),
and then form $R_{i+1}$ by replacing $C$ with~$Q$.
The other leaf clauses of~$Q$ have the
form $A^\prime$ or $\negU,A^\prime$,
the form $B^\prime$ or $\negU,B^\prime$,
the form $\negE^\prime,W$ or $\negU,\negE^\prime,W$,
and the form $\negA^\prime,\negB^\prime,E^\prime$.
By examination
of which variables are used for resolution in~$Q$,
it is clear that condition~d.{} is satisfied for
these clauses, and thus they are valid unfinished clauses
in~$R_{i+1}$.

This completes the description of~$R_{i+1}$ in this
case.  The clause~$D$ has become newly learned, and $<2^{3k}$
many new unfinished clauses have been introduced.

Now suppose $u=a$, so $\negU\in\negx_a^{\koplus}$.
(The case where $u=b$ is similar.) In this case,
the clause~$D$ has the form $\negU,\negB,E$ for
some $\negB$ in~$\negx_b^\koplus$ and
some some $E$ in~$x_e^\koplus$.
Let
$Q$ be the $\koplus$-translation of
\begin{prooftree}
\AxiomC{$\negU,\negx_b,x_e$}
\AxiomC{$x_b$}
\BinaryInfC{$\negU,x_e$}
\AxiomC{$\negx_e,W$}
\BinaryInfC{$\negU,W$}
\end{prooftree}
One of the leaf clauses in~$Q$ is equal to the
clause~$D$ to be learned, and $Q$ is ordered so that $D$ is
its leftmost leaf clause.  As before, $D$ is
replaced in~$Q$ with~(\ref{LearnDclauseEq}) and
then $R_{i+1}$ is formed by replacing $C$ in~$R_i$
with~$Q$.
The clause~$D$ has become learned as an input lemma,
and $<2^{2k}$ many new
unfinished clauses have been introduced in~$R_{i+1}$.

III.\ Third, consider the case where $C$ has the form $\negU,\negV,W$.
W.l.o.g., $v$~is not an ancestor of~$u$.  First
suppose $D$ is $E\in x_e^\koplus$ for $e$ a source in~$G$.
We again have $e$ is not equal to $u$, $v$, or~$w$.
Since also $e\in (G\rest w)[u,v]$ and since $u$ and~$v$ are
independent ancestors of~$w$, there must exist
a path from $w$ to~$e$ that avoids $u$ and~$v$,
a path from $w$ to~$u$ that avoids $e$ and~$v$, and
a path from $w$ to~$v$ that avoids $e$ and~$u$.
By looking at the point where these
three paths first diverge, there is a vertex~$f$
that lies on exactly two of these paths.  Suppose,
for instance, that $f$ lies on the two paths from $w$
to~$u$ and to~$e$. (The other two cases are similar.)
Then $u$ and~$e$ are independent ancestors of~$f$, and
$f$~and $v$ are independent ancestors of~$w$.  Form
the derivation structure
\begin{equation}\label{GPebProofEqDalpha}
\AxiomC{$x_e$}
\AxiomC{$\negU,\negx_e,x_f$}
\BinaryInfC{$\negU,x_f$}
\AxiomC{$\negV,\negx_f,W$}
\BinaryInfC{$\negU,\negV,W$}
\DisplayProof
\end{equation}
Let $Q$ be the $\koplus$-translation
of~(\ref{GPebProofEqDalpha}) with $D$ as its leftmost clause
and form $R_{i+1}$ as before.

Second suppose $D$ is $\negA,\negB,E$.
We assume that $e\not= w$,
leaving the slightly simpler $e=w$ case to the reader.  (See the earlier footnote.)
Suppose for the moment that $u$ and~$v$ are distinct from $a$ and~$b$.
As before, they cannot equal~$e$.
Consider the derivation structure
\begin{equation}\label{GPebProofEqC}
\AxiomC{$\negx_a,\negx_b, x_e$}
\AxiomC{$\qmark_1,x_a$}
\BinaryInfC{$\qmark_1,\negx_b,x_e$}
\AxiomC{$\qmark_2,x_b$}
\BinaryInfC{$\qmark_1,\qmark_2,x_e$}
\AxiomC{$\qmark_3,\negx_e,W$}
\BinaryInfC{$\negU,\negV,W$}
\DisplayProof
\end{equation}
There are a number of subcases to consider; in each we describe how to set
$\qmark_1$, $\qmark_2$, and $\qmark_3$.  By default,
$\qmark_1$, $\qmark_2$, and $\qmark_3$ (when not specified otherwise)
are to be the empty clause.
Some of the subcases are overlapping, and when so,
either option may be used.  For example, it can happen
that $u$ is both an ancestor of~$e$ and a member of $(G\rest w)[e]$.
\begin{description}
\item[$\bullet$] Suppose $u$ is an ancestor of~$e$ and
and $v\in (G\rest w)[e]$.
\begin{description}
\item[$\circ$] If $u$ is an ancestor of~$a$,
set $\qmark_1:= \negU$ and $\qmark_3:=\negV$.
\item[$\circ$] Otherwise, $u$ is an ancestor of~$b$, and we
set $\qmark_2:= \negU$ and $\qmark_3:=\negV$.
\end{description}
\item[$\bullet$] Suppose $v$ is an ancestor of~$e$,
and $u\in (G\rest w)[e]$.  This is handled like the previous
case, interchanging $\negU$ and~$\negV$.
\item[$\bullet$] Suppose neither $u$ nor~$v$ is in $(G\rest w)[e]$.
\begin{description}
\item[$\circ$] If $u$ is an ancestor of~$a$ and $v$~an ancestor of~$b$,
set $\qmark_1:=\negU$ and $\qmark_2:=\negV$.
\item[$\circ$] If $u$ is an ancestor of~$b$ and $v$~an ancestor of~$a$,
set $\qmark_1:=\negV$ and $\qmark_2:=\negU$.
\item[$\circ$] If $u$ and~$v$ are independent ancestors of~$a$,
set $\qmark_1:= \negU,\negV$.
\item[$\circ$] If $u$ and~$v$ are independent ancestors of~$b$,
set $\qmark_2:= \negU,\negV$.
\end{description}
These four subcases cover all possibilities since $u$ and~$v$ are
independent ancestors of~$w$, and there is no path
from $w$ to either $u$ or~$v$ which avoids~$e$.
\end{description}
If any of the above subcases hold, form~$Q$ as the $\koplus$-translation
of~(\ref{GPebProofEqC}) with $D$ as its leftmost clause,
and form $R_{i+1}$ by exactly the same
construction as in the earlier cases.

If, however neither $u$ nor~$v$ is an ancestor of~$e$,
then (\ref{GPebProofEqC}) cannot be used.  Since also
$e\in (G\rest w)[u,v]$, and $u$ and~$v$ are independent ancestors
of~$w$, there must be a path from $w$ to~$u$ that avoids
$v$ and~$e$, a path from $w$ to~$v$ that avoids
$u$ and~$e$, and a path from $w$ to~$e$ that avoids
$u$ and~$v$.  By looking at the point where these
three paths first diverge, there is a vertex~$f$
that lies on exactly two of these paths.  Suppose again
that $f$ lies on the two paths from $w$
to~$u$ and to~$e$. (The other two cases are similar.)
Then, $u$ and~$e$ are independent ancestors of~$f$, and
$f$~and $v$ are independent ancestors of~$w$.  Form
the derivation structure
\begin{equation}\label{GPebProofEqD}
\AxiomC{$\negx_a,\negx_b, x_e$}
\AxiomC{$x_a$}
\BinaryInfC{$\negx_b,x_e$}
\AxiomC{$x_b$}
\BinaryInfC{$x_e$}
\AxiomC{$\negU,\negx_e,x_f$}
\BinaryInfC{$\negU,x_f$}
\AxiomC{$\negV,\negx_f,W$}
\BinaryInfC{$\negU,\negV,W$}
\DisplayProof
\end{equation}
Let $Q$ be the $\koplus$-translation of~(\ref{GPebProofEqD})
with $D$ as its leftmost leaf clause.
Replace $D$ with~(\ref{LearnDclauseEq}), and then form $R_{i+1}$
by replacing with $C$ with~$Q$.
The other leaf clauses of~$Q$ including those
of the form
$\negV,\negF^\prime,W$
and of the form
$\negU, \negE^\prime,F^\prime$ become valid unfinished clauses
that satisfy condition~d.
The clause $D$ has become learned as an input lemma,
and $R_{i+1}$ has
$<2^{4k}$ new
unfinished clauses.

We still have to consider the cases where $u$, $v$, $a$
and~$b$ are not distinct.  There are four (very similar)
cases where only one of $u$ and~$v$ is in $\{a,b\}$.
For instance, suppose that $u=a$ and
$v\not= b$.  Then form $Q$ as the $\koplus$-translation
of
\begin{prooftree}
\AxiomC{$\negU,\negx_b,x_e$}
\AxiomC{$\qmark_1,x_b$}
\BinaryInfC{$\qmark_1,\negU,x_e$}
\AxiomC{$\qmark_2,\negx_e,W$}
\BinaryInfC{$\negU,\negV,W$}
\end{prooftree}
where $\qmark_1:=V$ if $v$ is an ancestor of~$b$, and $\qmark_2:=V$
if $v$ is not an ancestor of~$b$.   Order $Q$ so that $D$ is
its leftmost leaf clause, and form $R_{i+1}$ as before.

In the case where $u=a$ and $v=b$, or vice-versa,
let $Q$ instead be the $\koplus$-translation of
\begin{prooftree}
\AxiomC{$\negU,\negV,x_e$}
\AxiomC{$\negx_e,W$}
\BinaryInfC{$\negU,\negV,W$}
\end{prooftree}
and proceed as before.

This completes case~III.{} and
the construction of $R_{i+1}$ from~$R_i$.

In cases \caseiti-\caseitiii{}, an unfinished leaf is completely
handled without adding a new unfinished clause.
In cases~\caseitiv{}, a
new $\Peb^\koplus(G)$ clause is learned as an input lemma
while adding
fewer than $2^{4k}$ many new unfinished leaves.
There are fewer than $n 2^{3(k-1)}$ many clauses that can be
learned.
Therefore the process of forming~$R_i$'s
terminates with a refutation~$R$
after a finite number of steps. The refutation~$R$ is a valid
regRTI refutation of the $\GPeb^\koplus(G)$ clauses.

To estimate the size of~$R$, note further that
each time
an unfinished leaf is handled by cases \caseiti{}-\caseitiii{}
at most $k 2^{3k} n^2$ many clauses are added.   Therefore,
the refutation~$R$ has size
$O((2^{3k}n)(2^{4k})(k 2^{3k}n^2) = O(2^{11k} n^3)$.
The overall size of the $\GPeb^\koplus(G)$ clauses
is $\Omega( 2^{3k} n)$.  Thus the size of~$R$ is polynomially
bounded by the size of the $\GPeb^\koplus(G)$ principle; in fact,
it is bounded by a degree four polynomial.

This completes the proof of Theorem~\ref{GPebProofsThm}.
\hfill $\qed$
\end{proof}

\bibliographystyle{siam}
\bibliography{logic,cstheory}

\begin{thebibliography}{10}

\bibitem{AJPU:regularresolution}
{\sc M.~Alekhnovich, J.~Johannsen, T.~Pitassi, and A.~Urquhart}, {\em An
  exponential separation between regular and general resolution}, Theory of
  Computing, 3 (2007), pp.~81--102.

\bibitem{AFT:clauseLearning}
{\sc A.~Atserias, J.~K. Fichte, and M.~Thurley}, {\em Clause-learning
  algorithms with many restarts and bounded-width resolution}, Journal of
  Artificial Intelligence Research, 40 (2011), pp.~353--373.

\bibitem{BKS:clauselearning}
{\sc P.~Beame, H.~A. Kautz, and A.~Sabharwal}, {\em Towards understanding and
  harnessing the potential of clause learning}, J. Artificial Intelligence
  Research, 22 (2004), pp.~319--351.

\bibitem{BeckmannBuss:dLK}
{\sc A.~Beckmann and S.~R. Buss}, {\em Separation results for the size of
  constant-depth propositional proofs}, Annals of Pure and Applied Logic, 136
  (2005), pp.~30--55.

\bibitem{BenSasson:sizespace}
{\sc E.~Ben-Sasson}, {\em Size sapce tradeoffs for resolution}, SIAM Journal on
  Computing, 38 (2009), pp.~2511--2525.

\bibitem{BIW:nearoptimal}
{\sc E.~Ben-Sasson, R.~Impagliazzo, and A.~Wigderson}, {\em Near optimal
  separation of tree-like and general resolution}, Combinatorica, 24 (2004),
  pp.~585--603.

\bibitem{BonetBuss:poolVsRegularSAT}
{\sc M.~L. Bonet and S.~R. Buss}, {\em An improved separation of regular
  resolution from pool resolution and clause learning}, in Proc. 15th
  International Conference on Theory and Applications of Satisfiability Testing
  -- SAT 2012, Lecture Notes in Computer Science \#7317, 2012, pp.~45--57.

\bibitem{BonetBuss:poolVsRegularArxiv}
\leavevmode\vrule height 2pt depth -1.6pt width 23pt, {\em An improved
  separation of regular resolution from pool resolution and clause learning}.
\newblock Full version, arxiv.org, arXiv:1202.2296v2 [cs.LO], 2012.

\bibitem{BonetGalesi:ResAndPolyCalculus}
{\sc M.~L. Bonet and N.~Galesi}, {\em A study of proof search algorithms for
  resolution and polynomial calculus}, in 40th Annual IEEE Symp. on Foundations
  of Computer Science, IEEE Computer Society, 1999, pp.~422--431.

\bibitem{BussKolodziejczyk:SmallStone}
{\sc S.~Buss and L.~Ko{\l}odziejczyk}, {\em Small stone in pool}.
\newblock In preparation, 2012.

\bibitem{Buss:poolhard}
{\sc S.~R. Buss}, {\em Pool resolution is {NP}-hard to recognise}, Archive for
  Mathematical Logic, 48 (2009), pp.~793--798.

\bibitem{BHJ:ResTreeLearning}
{\sc S.~R. Buss, J.~Hoffmann, and J.~Johannsen}, {\em Resolution trees with
  lemmas: Resolution refinements that characterize {DLL}-algorithms with clause
  learning}, Logical Methods in Computer Science, 4, 4:13 (2008), pp.~1--18.

\bibitem{DLL:theoremproving}
{\sc M.~Davis, G.~Logemann, and D.~Loveland}, {\em A machine program for
  theorem proving}, Communications of the ACM, 5 (1962), pp.~394--397.

\bibitem{Goerdt:regularresolution}
{\sc A.~Goerdt}, {\em Regular resolution versus unrestricted resolution}, SIAM
  Journal on Computing, 22 (1993), pp.~661--683.

\bibitem{HBPvG:clauselearn}
{\sc P.~Hertel, F.~Bacchus, T.~Pitassi, and A.~Van~Gelder}, {\em Clause
  learning can effectively p-simulate general propositional resolution}, in
  Proc. 23rd AAAI Conf. on Artificial Intelligence (AAAI 2008), AAAI Press,
  2008, pp.~283--290.

\bibitem{HuangYu:regularresolution}
{\sc W.~Huang and X.~Yu}, {\em A {DNF} without regular shortest consensus
  path}, SIAM Journal on Computing, 16 (1987), pp.~836--840.

\bibitem{Johannsen:widthlearning}
{\sc J.~Johannsen}, {\em An exponential lower bound for width-restricted clause
  learning}, in Proc. 12th International Conference on Theory and Applications
  of Satisfiability Testing -- SAT 2009, Lecture Notes in Computer Science
  \#5584, 2009, pp.~128--140.

\bibitem{Krishnamurthy:trickyformulas}
{\sc B.~Krishnamurthy}, {\em Short proofs for tricky formulas}, Acta
  Informatica, 22 (1985), pp.~253--275.

\bibitem{PipatsrisawatDarwiche:clauselearning}
{\sc K.~Pipatsrisawat and A.~Darwiche}, {\em On the power of clause-learning
  {SAT} solvers as resolution engines}, Artificial Intelligence, 172 (2011),
  pp.~512--525.

\bibitem{SBI:SwitchingkDnf}
{\sc N.~Segerlind, S.~R. Buss, and R.~Impagliazzo}, {\em A switching lemma for
  small restrictions and lower bounds for {$k$}-{DNF} resolution}, SIAM Journal
  on Computing, 33 (2004), pp.~1171--1200.

\bibitem{Stalmarck:trickyformulas}
{\sc G.~St{\aa}lmarck}, {\em Short resolution proofs for a sequence of tricky
  formulas}, Acta Informatica, 33 (1996), pp.~277--280.

\bibitem{Urquhart:regularresolution}
{\sc A.~Urquhart}, {\em A near-optimal separation of regular and general
  resolution}, SIAM Journal on Computing, 40 (2011), pp.~107--121.

\bibitem{VanGelder:PoolResolution}
{\sc A.~Van~Gelder}, {\em Pool resolution and its relation to regular
  resolution and {DPLL} with clause learning}, in Logic for Programming,
  Artificial Intelligence, and Reasoning (LPAR 2005), Lecture Notes in Computer
  Science 3835, Springer-Verlag, 2005, pp.~580--594.

\bibitem{VanGelder:InputCoverNumber}
\leavevmode\vrule height 2pt depth -1.6pt width 23pt, {\em Preliminary report
  on input cover number as a metric for propositional resolution proofs}, in
  Theory and Applications of Satisfiability Testing - SAT 2006, Lecture Notes
  in Computer Science 4121, Springer Verlag, 2006, pp.~48--53.

\end{thebibliography}

\end{document}